\documentclass[runningheads]{llncs}

\usepackage[T1]{fontenc}
\usepackage[utf8]{inputenc}
\usepackage[english]{babel}
\usepackage[babel]{csquotes}

\usepackage{amsmath}
\usepackage{amssymb}
\usepackage{mathtools}
\usepackage{stmaryrd}
\DeclareMathAlphabet{\mathpzc}{OT1}{pzc}{m}{it}

\usepackage{wrapfig}
\usepackage{mathcommand}

\usepackage{stmaryrd}
\SetSymbolFont{stmry}{bold}{U}{stmry}{m}{n}
\usepackage{mathrsfs}
\DeclareMathAlphabet{\mathpzc}{OT1}{pzc}{m}{it}
\usepackage{bm}

\DeclareFontFamily{U} {MnSymbolA}{}
\DeclareFontShape{U}{MnSymbolA}{m}{n}{
  <-6> MnSymbolA5 
  <6-7> MnSymbolA6 
  <7-8> MnSymbolA7 
  <8-9> MnSymbolA8 
  <9-10> MnSymbolA9 
  <10-12> MnSymbolA10 
  <12-> MnSymbolA12}{} 
\DeclareFontShape{U}{MnSymbolA}{b}{n}{
  <-6> MnSymbolA-Bold5 
  <6-7> MnSymbolA-Bold6 
  <7-8> MnSymbolA-Bold7 
  <8-9> MnSymbolA-Bold8 
  <9-10> MnSymbolA-Bold9 
  <10-12> MnSymbolA-Bold10 
  <12-> MnSymbolA-Bold12}{} 
\DeclareSymbolFont{MnSyA} {U} {MnSymbolA}{m}{n}
\DeclareMathSymbol{\rsim}{\mathrel}{MnSyA}{160}
\DeclareMathSymbol{\lsim}{\mathrel}{MnSyA}{162}

\usepackage{xspace}

\usepackage{enumerate}

\usepackage{orcidlink}

\usepackage[dvipsnames,table]{xcolor}
\usepackage{float}
\usepackage{subfig}
\usepackage{caption}

\usepackage{tikz}
\usetikzlibrary{
  automata,
  cd,
  positioning,
  decorations.pathreplacing,
  calligraphy,
  calc,
  bbox
}
\tikzset{
  >=stealth,
  accepting distance=0.3cm,
  initial text=,
  initial distance=0.3cm,
  every state/.style = {minimum size=0.4cm}
}
\tikzcdset{arrow style=math font}

\usepackage{tabularx}

\usepackage{hyperref}
\usepackage{color}

\newcommand{\alp}{\Sigma}
\newcommand{\alpp}{\Gamma}
\newcommand{\wds}[1][\alp]{{#1}^{*}}
\newcommand{\ltr}{\sigma}
\newcommand{\e}{\epsilon}

\NewDocumentCommand{\eqCl}{om}{
\IfNoValueTF{#1}
{[#2]}
{[#2]_{#1}}
}

\LoopCommands{\lettersUppercase}[#1cal]{\declaremathcommand{#2}{\mathcal{#1}}}
\LoopCommands{\lettersUppercase}[#1bb]{\declaremathcommand{#2}{\mathbb{#1}}}

\newcommand{\NFA}{\textup{\textrm{NFA}}\xspace}
\newcommand{\DFA}{\textup{\textrm{DFA}}\xspace}
\newcommand{\SST}{\textup{\textrm{aSST}}\xspace}

\newcommand{\eg}{\textit{e.g.~}}
\newcommand{\ie}{\textit{i.e.~}}

\newcommand{\resp}{resp.~}
\newcommand{\wrt}{w.r.t.~}

\newcommand{\set}[2][]{#1\{ #2 #1\}}
\newcommand{\setWCnd}[3][]{#1\{ #2 \,#1\lvert\, #3 #1\}}
\newcommand{\intInterv}[3][]{#1\llbracket #2, #3 #1\rrbracket}
\newcommand{\card}[2][]{#1\lvert #2 #1\lvert}
\newcommand{\lenWd}[2][]{#1\lvert #2 #1\lvert}

\newcommand{\lanNFA}[2][]{#1\llbracket #2 #1\rrbracket}
\newcommand{\fctTrn}[2][]{#1\llbracket #2 #1\rrbracket}
\newcommand{\relTrn}[2][]{#1\llbracket #2 #1\rrbracket}

\newcommand{\fctBim}[2][]{#1\llbracket #2 #1\rrbracket}
\newcommand{\fctABim}[2][]{#1\llbracket #2 #1\rrbracket}
\newcommand{\fctSST}[2][]{#1\llbracket #2 #1\rrbracket}

\newcommand{\rat}[2][]{\mathrm{Rat}#1( #2 #1)}

\newcommand{\fct}[2]{\colon #1 \to #2}

\newcommand{\dist}[3][]{#1\lVert #2, #3 #1\rVert}

\newcommand{\oFct}{\mathpzc{o}}
\newcommand{\iniOFct}{\mathpzc{i}}
\newcommand{\finOFct}{\mathpzc{f}}

\newcommand{\varSet}{\mathcal{X}}
\newcommand{\expr}[2][]{\mathrm{Exp}#1( #2 #1)}
\newcommand{\subs}[2][]{\mathrm{Sub}#1( #2 #1)}
\newcommand{\valu}[2][]{\mathrm{Val}#1( #2 #1)}

\newcommand{\outFctSST}{\gamma}

\newcommand{\sttTrSST}{\delta^{\mathrm{s}}}
\newcommand{\regUpSST}{\delta^{\mathrm{r}}}

\newcommand{\rSST}{k}
\newcommand{\sSST}{n}
\newcommand{\rCRA}{\rSST}
\newcommand{\sCRA}{\sSST}

\newcommand{\lOutFctBim}{\lambda}
\newcommand{\rOutFctBim}{\rho}
\newcommand{\outFctBim}{\omega}
\newcommand{\lIni}{l_{\mathrm{i}}}
\newcommand{\rFin}{r_{\mathrm{f}}}

\renewcommand{\emptyset}{\varnothing}

\newcommand{\regOut}{X_{\mathrm{out}}}

\newcommand{\congDFA}[1]{\sim_{#1}}
\newcommand{\DFACong}[1]{\Acal_{#1}}

\NewDocumentCommand{\outSeqTr}{o}{
    \IfNoValueTF{#1}
    {*}
    {*_{#1}}
}
\DeclareMathOperator{\dom}{dom}

\newcommand{\PbRegMin}{\textup{\textsc{Reg Min}}\xspace}
\newcommand{\PbFARegMin}{\textup{\textsc{FA-Reg Min}}\xspace}
\newcommand{\PbSttRegMin}{\textup{\textsc{Stt-Reg Min}}\xspace}

\newcommand{\p}{\textup{\textsc{P}}\xspace}
\newcommand{\np}{\textup{\textsc{NP}}\xspace}
\newcommand{\npc}{\textup{\textsc{NP-complete}}\xspace}

\newcommand{\ptime}{\textup{\textsc{PTime}}\xspace}

\newcommand{\pspace}{\textup{\textsc{PSpace}}\xspace}

\newcommand{\trn}[2]{#1 \left|\,\begin{array}{@{}l@{}}#2\end{array}\right.}

\NewDocumentCommand{\tikzFig}{mmo}{\IfNoValueTF{#3}{
    \begin{figure}
      \centering\input{./figures/tikz/#1.tex}\caption{#2}\label{fig:#1}
    \end{figure}
  }{
    \begin{figure}
      \centering\input{./figures/tikz/#1.tex}\caption{#2}\label{fig:#3}
    \end{figure}
  }}

\usepackage[T1]{fontenc}
\usepackage{graphicx}

\newsavebox{\largest}

\begin{document}
\title{Minimizing Streaming String Transducers:\\ An algebraic approach}
%
%
\author{
  Yahia Idriss Benalioua\inst{1}\orcidlink{0009-0003-3980-1315}
  \and
  Nathan Lhote\inst{2}\orcidlink{0000-0003-3303-5368}
  \and
  Pierre-Alain Reynier\inst{2}\orcidlink{0009-0008-4345-704X}
}
\authorrunning{Y. I. Benalioua et al.}
%
\institute{
  University of Warsaw, Poland \email{y.benalioua@uw.edu.pl}
  \and
  Aix Marseille Univ, CNRS, LIS, Marseille, France \email{\{nathan.lhote,pierre-alain.reynier\}@lis-lab.fr}
}
\maketitle              

\begin{abstract}
In this work, we study minimization of rational functions
given as appending streaming string transducers (\SST for short). 
We rely on 
an algebraic presentation of these functions,
known as bimachines, to address
 the minimization of both states and registers of \SST.
  
First, we show a bijection between a subclass of \SST and bimachines, 
which maps the numbers of states and registers of the \SST to two natural parameters of the bimachine. Using known results on the minimization
of bimachines, this yields a \ptime (\resp \np) 
procedure to minimize this subclass of \SST
with respect to registers (\resp both states and registers). 
In a second step, we introduce a new model of bimachines, 
named asynchronous bimachines, which allows to 
lift the bijection to the whole class of \SST.
Based on this, we prove that register minimization with a fixed 
underlying automaton is \npc.

\end{abstract}
\section{Introduction}%
\label{sec:intro}

Important connections between computation, mathematical logic, and algebra have been
established for regular languages of finite words. A language is regular 
iff it is recognized by a finite automaton, iff it is definable in
monadic second-order logic (MSO) with one successor~\cite{Bu60,Elg61,tra61short}, and iff its syntactic monoid, a canonical monoid attached
to every language, is finite (see for instance~\cite{str94}). While automata are well-suited to study the algorithmic properties of regular languages, the algebraic view has
provided effective characterizations of regular languages and its subclasses. Most notably,
the problem of deciding whether a regular language is first-order definable amounts to
checking whether its syntactic monoid, which is computable from any finite automaton
recognizing the language, is aperiodic, which is decidable.

At the computational level, word functions are defined by (one-way) transducers, which extend
automata with outputs on their transitions. The non-determinism of automata may result in
mapping an input to different output words. This yields the class of rational relations, long studied in the literature, but for which numerous problems are undecidable, starting with equivalence. When the relation is functional, we obtain the class
of \emph{rational functions}~\cite{BerstelB79}, which behaves much better regarding decidability. 

While non-determinism is important for expressiveness, it is not realistic regarding actual 
implementations. An alternative model, known as \emph{streaming string
 transducers}~\cite{AluCer2010,AluCer2011}, has
been introduced to remedy this problem. The underlying automaton is deterministic,
but the model is no longer finite-state: it is equipped with finitely many registers that are
used to store partial output words, and that can be updated during the execution.

For automata based models, a very important problem is to simplify the models. For instance, deterministic machines are essential
in order to derive efficient evaluation algorithms. Similarly, reducing the size of the models allows to reduce the computation time of
most algorithms.
For the class of regular languages, classical algebraic tools such as Myhill-Nerode congruence not only allow to minimize automata, but also yield canonical models, which can be used to derive learning algorithms.
When considering automata models extended with registers, the minimization may be with respect to 
the number of registers, the number of states, or both. For instance, this has been studied
for cost register automata over a field in~\cite{BenLhoRey2024a}.

In this work, we consider the class of rational functions, and are interested in minimization problems for streaming string transducers. In~\cite{DavReyTal2016}, the problem of minimizing the number of registers has been
solved for this class, and shown to be \pspace-complete. More precisely, the approach
is based on the presence of some pattern in a transducer realizing the function, and, equivalently,
on some machine independent property of the function itself.
This result allows to identify the minimal number of registers needed to realize the function,
but comes with a blow-up of the number of states.

Our objective is to follow an algebraic approach to address minimization
of streaming string transducers, \wrt both the number of states and registers.
 Indeed, Reutenauer and Schutzenberger
have introduced in~\cite{ReuSch1991} congruences that allow to derive an algebraic characterization
of rational functions, and a representation of these functions by an alternative model known as
bimachine (the name comes from the fact that it is represented as a pair of automata with outputs, known as a left and a right automaton). Unlike other models, this algebraic characterization does not yield a unique, but a finite family of canonical bimachines.
Our contributions can be summarized as follows: (see also the table)
\begin{enumerate}
\item We prove a strong equivalence between bimachines and a subclass of \SST (denoted $\SST_{\textsf{iffo}}$), which maps the number of registers (\resp states) of the \SST to the size of the right (\resp left) automaton of the bimachine. Using known results on the minimization of bimachines, we can minimize the corresponding subclass of \SST with respect to registers in \ptime, and with respect to states and registers in \np.
\item We introduce a new model of bimachines, named asynchronous bimachines, which allows to lift the previous equivalence to the whole class of \SST. Based on this, we consider register minimization with fixed underlying automaton, and show 
that (a slight generalization of) this problem is \npc.
\end{enumerate}

Due to lack of space, omitted proofs can be found in the Appendix.

\begin{center}
\begin{tabular}{l|c|c|c}
 & \shortstack{Fixed Underlying Automaton\\ Register Minimization} & \shortstack{Register\\ Minimization} & \shortstack{State-Register\\ Minimization} \\
 \hline
\SST & \shortstack{{\bf \np-complete}\\ (for \SST with partial updates)} & \pspace-complete~\cite{DavReyTal2016} & open \\
\hline
$\SST_{\textsf{iffo}}$ &{\bf \ptime }& {\bf \ptime }& {\bf \np}
\end{tabular}
\end{center}

\section{Models of Transducers}%
\label{sec:models}

\paragraph{Words, languages}

An alphabet $ \alp$ is a finite set.
An element $\ltr \in \alp$ is called a \emph{letter} and a finite sequence of
letters is called a \emph{word}.
The length of a word $w$ will be denoted by $ \lenWd{w}$.
The word of length 0 is called the empty word and is denoted by $ \e$.
The number of occurrences of a letter $\ltr$ in a word $w$ will be denoted by
$|w|_\ltr$.

Given two words $w = a_1 \dots a_n$ and $w' = a'_1 \dots a'_m$,
we will denote by $w \cdot w'$ or simply $ww'$ their \emph{concatenation}
$a_1 \dots a_n a'_1 \dots a'_m$ and the concatenation of $w$ with itself $n$
times will be denoted by $w^n$ ($w^0 = \e$ for all $w$).

The set of all words over $\alp$ is denoted by $ \wds$.
Together with concatenation, $\wds$ is a monoid called the free monoid
over $\alp$.
Its neutral element is the empty word $\e$.
Any subset $L \subseteq \wds$ is called a \emph{language}.
The class of rational (or regular) languages over $\alp$ is denoted
$\rat{\wds}$.

Let $u$ and $v$ be two words of $\wds$. We will write $u \leq v$ if $u$ is a prefix of $v$, \ie if there exists $w \in \wds$
  such that $v = uw$.
  We will denote $w$ by $u^{-1}v$.
  The longest common prefix of $u$ and $v$ is denoted by $u \wedge v$.
  We generalize this definition to languages in the natural way and define
  $\bigwedge L \in \wds$ as the longest common prefix of all the words of $L$
  for all nonempty language $L \subset \wds$ and $\bigwedge \emptyset = \e$.

\paragraph{Automata}


A finite automaton $\Acal$ on an alphabet $\alp$ is a tuple $(Q, I, \Delta, F)$
where $Q$ is a finite set of \emph{states}, $I \subseteq Q$ and $F \subseteq Q$
are the sets of \emph{initial} and \emph{final} states respectively and
$\Delta \subseteq Q \times \alp \times Q$ is the set of \emph{transitions}
of the automaton.

A \emph{run} of $\Acal$ on a word $w = a_1 \dots a_n \in \wds$ is a finite sequence
of states $(q_i)_{i \in \intInterv{0}{n}}$ such that $(q_{i-1}, a_i, q_i) \in \Delta$
for all $i \in \intInterv{1}{n}$.
It will be denoted by $q_0 \xrightarrow{w}_\Acal q_n$ or $q_0 \xrightarrow{w} q_n$
when $\Acal$ is clear from the context.

We say that the run is \emph{accepting} if $q_0 \in I$ and $q_n \in F$
and define the language \emph{recognized} by $\Acal$, denoted by $\lanNFA{\Acal}$,
as the set of words for which there exists an accepting run of $\Acal$.

$\Acal$ is called \emph{deterministic} (\resp \emph{complete}) if, for all $p \in Q$
and $\ltr \in \alp$, there exists at most (\resp at least) one state $q \in Q$
such that $(p,\ltr,q) \in \Delta$.
A deterministic automaton is also required to have a single initial state.
We say that $\Acal$ is \empty{codeterministic}  if the mirror automaton of $\Acal$, obtained
by reversing transitions and swapping initial and final states,
is deterministic.
Last, $\Acal$ is \emph{unambiguous} if,
for all $w \in \wds$, there exists at most $1$ accepting run of $\Acal$ on $w$.

In the following, we will use the abbreviations \NFA and \DFA for
\emph{nondeterministic finite automaton}
and \emph{deterministic finite automaton} respectively.

Given a complete \DFA $\Acal$, its transitions induce an action
from words on states, that we denote by $\cdot_\Acal$, which, given a state
$q$ and a word $w$, returns $q \cdot_\Acal w$ the unique state reached from
$q$ by reading the word $w$. 
It also induces a right congruence on $\wds$, denoted by $\congDFA{\Acal}$,
defined by $u \congDFA{\Acal} v$ if and only if $q_0 \cdot_\Acal u = q_0 \cdot_\Acal v$,
where $q_0$ is the initial state of $\Acal$.
Conversely, given a right congruence of finite index $\sim$ on $\wds$ and choosing
a set $F$ of equivalence classes, we can define a \DFA $\DFACong{\sim}$
with $\wds /\sim$ as its set of states, $\eqCl[\sim]{\e}$ as its initial state,
$F$ as its set of final states, and, for all $u \in \wds$ and $\ltr \in \alp$,
$\eqCl[\sim]{u} \cdot_{\DFACong{\sim}} \ltr = \eqCl[\sim]{u\ltr}$.

Given two binary relations $R$ and $R'$ on $\wds$, we say that $R$ \emph{finer}
than $R'$ (or $R'$ is coarser than $R$) whenever, for all $u,v \in \wds$, if
$u R v$ then $u R' v$.
Moreover, a \DFA $\Acal$ is called \emph{finer} than a relation $R$ whenever its associated
congruence $\congDFA{\Acal}$ is finer than $R$.
In particular, $\Acal$ is called \emph{finer} than another \DFA $\Acal'$ whenever
$\congDFA{\Acal}$ is finer than $\congDFA{\Acal'}$.
Intuitively, in this case, there exist a morphism between the states of $\Acal$
and those of $\Acal'$ \enquote{merging} states together.

All the definitions above are adapted as expected when $\Acal$ is codeterministic and co-complete.

\vspace{-.5cm}
\begin{figure}[h]
  \savebox{\largest}{\scalebox{.9}{
\begin{tikzpicture}[auto, initial text=]
  \node[state, accepting below, initial above] at (0,0) (v0) {};
  \node[state, initial] (v1) at (-1.2,0){};
  \node[state, initial right] (v2) at (1.2,0) {};
  \path[->] (v1) edge node[swap] {$\trn{b}{b}$} (v0);
  \path[->] (v1) edge[loop above] node[align=left] {$\trn{a}{b}$} ();
  \path[->] (v1) edge[loop below] node[align=left] {$\trn{b}{b}$} ();
  \path[->] (v2) edge node {$\trn{a}{a}$} (v0);
  \path[->] (v2) edge[loop above] node[align=left] {$\trn{a}{a}$} ();
  \path[->] (v2) edge[loop below] node[align=left] {$\trn{b}{a}$} ();
\end{tikzpicture}
}}
  \centering
  \mbox{}\hfill
  \subfloat[Functional transducer]{
    \usebox{\largest}\label{fig:FST_last}
  }\hfill
  \subfloat[Streaming String Transducer]{
    \raisebox{\dimexpr.5\ht\largest-.5\height}{
      \tikzset{parallel/.style={auto=right,->,
      to path={ let \p1=(\tikztostart),\p2=(\tikztotarget),
          \n1={atan2(\y2-\y1,\x2-\x1)},\n2={\n1+180}
          in ($(\tikztostart.{\n1})!1mm!270:(\tikztotarget.{\n2})$) --
          ($(\tikztotarget.{\n2})!1mm!90:(\tikztostart.{\n1})$) \tikztonodes}}}

 \scalebox{.9}{
\begin{tikzpicture}[auto, node distance=3cm]
  \node[state,
    initial, initial text = $\trn{}{X \coloneqq \e\\ Y \coloneqq \e}$,
    accepting below, accepting text = $X$
  ] (0) {};
  \node[state] (1) [right=of 0, accepting below, accepting text = $Y$] {};

  \path[->]
    (0) edge[parallel] node[align=left] {$\trn{b}{X \coloneqq Xa\\ Y \coloneqq Yb}$} (1)
    (1) edge[parallel] node[align=left] {$\trn{a}{X \coloneqq Xa\\ Y \coloneqq Yb}$} (0)
    (0) edge[loop above] node[align=left] {$\trn{a}{X \coloneqq Xa\\ Y \coloneqq Yb}$} (0)
    (1) edge[loop above] node[align=left] {$\trn{b}{X \coloneqq Xa\\ Y \coloneqq Yb}$} (1)
  ;
\end{tikzpicture}
 }}\label{fig:SST_indep}
  }\hfill\mbox{}\caption{A functional transducer and an equivalent \SST}\label{fig:last}
\end{figure}
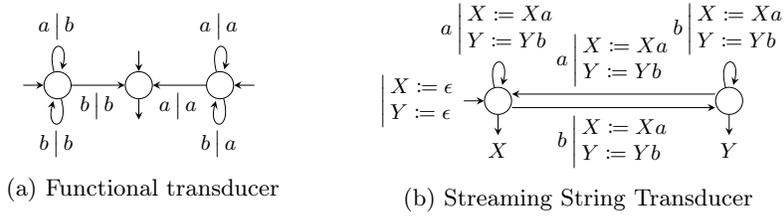

\vspace{-1cm}

\paragraph{Transducers}
  are defined over an input ($\alp$) and an output ($\alpp$) alphabet.
\begin{definition}
%
  A finite-state transducer from $\wds$ to $\wds[\alpp]$ is a tuple $\Tcal =
    (Q, \iniOFct, \oFct, \finOFct)$ where $Q$ is a finite set of \emph{states}
  and $\iniOFct \fct{Q}{\wds[\alpp]}$, $\oFct \fct{Q \times \alp \times Q}{\wds[\alpp]}$ 
  and $\finOFct \fct{Q}{\wds[\alpp]}$ are partial
  functions called respectively the \emph{initial output}, \emph{output}
  and \emph{final output} functions of $\Tcal$.
  \end{definition}

The \emph{underlying automaton} of $\Tcal$ is defined as the \NFA
$\Acal = (Q, I, \Delta, F)$ where $I = \dom(\iniOFct)$, $\Delta = \dom(\oFct)$
and $F = \dom(\finOFct)$.
$\Tcal$ is said \emph{sequential} (\resp \emph{cosequential}) if $\Acal$ 
is a deterministic (\resp codeterministic).

A run of $\Tcal$ on a word $w \in \wds$ is a run
$\rho = (q_i)_{i \in \intInterv{0}{n}}$ of $\Acal$ on $w$ with transitions
$(q_{i-1}, a_i, q_i) \in Q \times \alp \times Q$ for all
$i \in \intInterv{1}{n}$, such that $w = a_1 \cdots a_n$, together with an
\emph{output} word of $\wds[\alpp]$ defined by the concatenation
\[
  \oFct(\rho) = \iniOFct(q_0)\,
  \oFct(q_0, a_1, q_1) \cdots \oFct(q_{n-1}, a_n, q_n)\,
  \finOFct(q_n)
\]

$\Tcal$ \emph{realizes} a relation $ \relTrn{\Tcal} \subseteq \wds \times \wds[\alpp]$
defined by
\[
   \relTrn{\Tcal} = \setWCnd{(u,v) \in \wds \times \wds[\alpp]}{\exists
    \text{ a run } \rho \text{ of } \Tcal \text{ on } u \text{ with output }
    \oFct(\rho) = v}
\]


In particular, when $ \fctTrn{\Tcal}$ is a partial function, $\Tcal$ is called a
\emph{functional} transducer. This defines the class of \emph{rational functions},
as those realizable by functional transducers.
It is worth noting that a function is realizable by a functional transducer if
and only if it is realizable by an unambiguous one (see \eg\cite[Theorem IV.4.2]{Ber1979}).



  \begin{example}\label{ex:FST_last}
    Let $\alp = \alpp = \set{a,b}$.
    The functional transducer depicted on Figure~\ref{fig:FST_last} realizes a 
    function mapping each word $w \in \wds$ to $\ltr^{\lenWd{w}}$ where $\ltr$ 
    is the last letter of $w$ ($\e$ is mapped to itself). 
    It is, in particular, a cosequential transducer.
  \end{example}

\paragraph{Streaming String Transducers}
have originally
been introduced with general updates allowing to capture the
class of regular functions~\cite{AluCer2010}. However, the subclass of
\emph{appending} streaming string transducers (\SST) 
were shown to have the same expressiveness as
rational functions  (see \eg\cite[Theorem 4]{AluDADesRagYua2013}).
Hence, in the following, we will focus on \SST.
We will not elaborate on SST with more general register updates, but the interested
reader can find characterizations of their expressiveness in~\cite{AluCer2010}
and~\cite{FilRey2021} (see also~\cite[Section 5]{FilRey2016}).

\begin{definition}[Appending expressions]
Given a finite set of variables $\varSet$ and an output alphabet
$\alpp$, consider the set of expressions $\expr{\varSet}$ of the form\footnote{We will usually omit the
    $\cdot$ from the expressions and omit the $w$ when $w=\e$.} 
$X\cdot w$, where $X \in \varSet$ and $w \in \alpp^*$.
\end{definition}

A substitution over $\varSet$ is a map $s \fct{\varSet}{\expr{\varSet}}$.
It can be extended into a map from $\expr{\varSet}$ to $\expr{\varSet}$ by substituting
each variable $X$ in the expression given as an input by $s(X)$.
We can compose substitutions by identifying them with their extension.
Valuations are maps of the form $v \colon \varSet \to \alpp^*$, which
can also be extended to $\expr{\varSet}$.
 The set of substitutions will be denoted by $\subs{\varSet}$,
  the set of valuations by $ \valu{\varSet}$ and, given a substitution
  $s$ and a variable $X$, we will write $X \coloneqq e$ whenever $s(X) = e$ for
  some expression $e$.

\begin{definition}[Appending Streaming String Transducer]\label{def:SST}
  An appending streaming string transducer from $\wds$ to $\wds[\alpp]$ is a tuple
  $\Scal = (Q, \varSet, q_0, v_0, \sttTrSST, \regUpSST, \outFctSST)$ where
    $Q$ is a finite set of \emph{states},
    $\varSet$ is a finite set of \emph{registers},
    $q_0 \in Q$ is the \emph{initial} state,
    $v_0 \in \valu{\varSet}$ is the registers' \emph{initial valuation},
    $\sttTrSST \fct{Q \times \alp}{Q}$ is the \emph{transition function},
    $\regUpSST \fct{Q \times \alp}{\subs{\varSet}}$ is the \emph{register update function},
    and $\outFctSST \fct{Q}{\expr{\varSet}}$ is a partial \emph{output function}.
\end{definition}

We define the \emph{underlying automaton} of $\Scal$ as the \DFA $\Acal$ with
the same set of states, initial state and transition function as $\Scal$.
Its set of final states is $\dom(\outFctSST)$.

The runs of $\Scal$ will be runs of $\Acal$ together with valuations mapping
each register to its current value.
\ie we define the configurations of $\Scal$ as pairs $(q,v) \in Q \times \valu{\varSet}$.
Given a word $w = a_1 \dots a_n \in \wds$, the run of $\Scal$ on $w$ is the
sequence of configurations $(q_i,v_i)_{i \in \intInterv{0}{n}}$
where $q_0$ is the initial state, $v_0$ is the initial valuation and, for all
$i \in \intInterv{1}{n}$, $q_i = \sttTrSST(q_{i-1}, a_i)$ and
$v_i = v_{i-1} \circ \regUpSST(q_{i-1},a_i)$.
The function $\fctSST{\Scal} \fct{\wds}{\wds[\alpp]}$ \emph{realized} by
$\Scal$ is defined by
$\fctSST{\Scal}(w) = v_n (\outFctSST(q_n))$.

\begin{lemma}[\cite{AluDADesRagYua2013}]
  Let $f \fct{\wds}{\wds[\alpp]}$ be a partial function.
$f$ is a rational function iff it can be realized by an \SST.
\end{lemma}
Intuitively, from a transducer to an \SST, one uses one register per state
and performs a classical subset construction. Conversely, for each register,
one builds a run that simulates its value. This leads to an unambiguous transducer.

We will consider an extension of this model, in which
we allow the register update function to output
partial substitutions, \ie partial maps from $\varSet$ to $\expr{\varSet}$.
We call these machines \emph{\SST with partial updates}.
The semantics of \SST can be lifted to this model by considering
partial valuations.

\begin{example}
  The \SST depicted on Figure~\ref{fig:SST_indep} realizes the same function as 
  the transducer of Example~\ref{ex:FST_last} depicted on Figure~\ref{fig:FST_last}.
\end{example}

\paragraph{Problems considered}

The first problem we consider, which has been studied before
in~\cite{DavReyTal2016},
is the register minimization problem. It can be stated as follows:

\noindent
\textbf{Problem:} Register Minimization (\PbRegMin for short)\\
\noindent \textit{Input:} a rational function $f$ and $k\in \mathbb{N}$\\
\noindent \textit{Question:} Does there exist an \SST $\Scal$ that realizes $f$ and has $k$ registers?

\smallskip
A more constrained version, which has been studied in~\cite{BenLhoRey2024a} for
Cost-register automata over a field, can be defined as follows:

\noindent
\textbf{Problem:} State-Register Minimization (\PbSttRegMin for short)\\
\noindent \textit{Input:} a rational function $f$ and $k,n\in \mathbb{N}$\\
\noindent \textit{Question:} Does there exist an \SST $\Scal$ that realizes $f$ and has $k$ registers and $n$ states?

\smallskip Let $f$ be a rational function.
The \emph{register complexity} (\resp \emph{state complexity})
of $f$ is the minimal
number of registers (\resp states) needed for an \SST to realize it.
The \emph{state-register complexity} of $f$ is the set of pairs $(n,k)$ such that
there exists an \SST with $n$ states and $k$ registers realizing $f$, and such that any \SST realizing $f$ with $n'$ states and $k'$ registers with 
$(n,k)\neq (n',k')$ has either strictly more states ($n'>n$) or strictly more registers
($k'>k$).


\section{Algebraic Toolbox}%
\label{sec:algebra}
Let us now consider functional transducers from $\wds$ to $\wds[\alpp]$ for two
finite alphabets $\alp$ and $\alpp$, and introduce the tools that allow to obtain 
canonical and minimal representations of the functions they realize.

\paragraph{Bimachines}
Introduced by Schützenberger in~\cite{Sch1961b} and named
by Eilenberg in~\cite{Eil1974}, \emph{bimachines} read their input simultaneously 
from left to right using their left automaton and in the opposite direction using 
their right automaton and produces the output based on the states they reach during 
the run. 
They realize the same functions as functional transducers and are defined more 
formally as follows:

\begin{definition}[Bimachine]\label{def:bim}
    A bimachine on two finite alphabets $\alp$ and $\alpp$ is a tuple
  $\Bcal = (\Lcal, \Rcal, \lOutFctBim, \outFctBim, \rOutFctBim)$ where
    $\Lcal = (L, \lIni, \delta_\Lcal, F_\Lcal)$ is the \emph{left
            automaton} of $\Bcal$.
          It is a deterministic finite automaton.
    Dually, $\Rcal = (R, I_\Rcal, \delta_\Rcal, \rFin)$ is its
          \emph{right automaton}.
          It is a codeterministic finite automaton.
    $\lOutFctBim \fct{I_\Rcal}{\wds[\alpp]}$, $\outFctBim
            \fct{L \times \alp \times R}{\wds[\alpp]}$, and
          $\rOutFctBim \fct{F_\Lcal}{\wds[\alpp]}$ are respectively its
          \emph{left output}, \emph{output} and \emph{right output} functions.
\end{definition}

Defining $\outFctBim (l, \e, r) = \e$ and using the identity 
$\outFctBim (l, uv, r) = \outFctBim(l,u, v \cdot_\Rcal r)\cdot \outFctBim(l \cdot_\Lcal u, v,r)$,
for all $u,v \in \wds$, $l \in L$ and $r \in R$, $\outFctBim$ is extended to words by induction. 
The function $ \fctBim{\Bcal} \fct{\wds}{\wds[\alpp]}$ \emph{realized} by
$\Bcal$ is then defined, for all $w \in \wds$, by
$   \fctBim{\Bcal}(w) = \lOutFctBim(w \cdot_\Rcal \rFin)\, \outFctBim (\lIni, w, \rFin)
  \, \rOutFctBim(\lIni \cdot_\Lcal w)
$.
See Appendix~\ref{apx:ex-bimachine} for an example.

Any functional transducer can be converted into a bimachine and vice
versa~\cite[Theorem XI.7.1]{Eil1974} (see also~\cite[Theorem IV.5.1]{Ber1979}).

\begin{theorem}
  A partial function $f \fct{\wds}{\wds[\alpp]}$ is rational
  if and only if there exists a bimachine $\Bcal$ such that $\fctBim{\Bcal} = f$.
\end{theorem}

In contrast with \DFA, rational functions does not always have a unique minimal 
transducer realizing them.
However, Reutenauer and Schützenberger showed in~\cite{ReuSch1991} that they can 
still be represented in a canonical way using bimachines. 

\paragraph{Syntactic congruences}

Given a language $L\subseteq \wds$, and a word $u\in \wds$,
we define the left residual of $L$ \wrt $u$ as
the language $u^{-1}L = \{v\in \wds \mid uv \in L\}$.
Right residuals are defined symmetrically.
We define the equivalence relation $\sim_L$ as
$u\sim_L v $ iff $u^{-1}L = v^{-1}L$.
It is well known that $\sim_L$ is a right congruence, called the \emph{Myhill-Nerode congruence} of $L$,
and that  $L$ is a regular language iff
$\sim_L$ has finite index $k$. In addition, when this holds, 
the minimal \DFA recognizing $L$
has precisely $k$ states.
We consider the distance on words defined, for all $u,v \in \wds$, by $\dist{u}{v}
  = \lenWd{u} + \lenWd{v} -2 \lenWd{u \wedge v}$.
Two words are close for this distance if they share a long common prefix.
Based on this distance, congruences characterizing rational transductions were 
introduced in~\cite{ReuSch1991}.

\begin{definition}\label{def:synConBim}
  Let $f \fct{\wds}{\wds[\alpp]}$ be a partial function.

  The \emph{left syntactic congruence} $\lsim_f$ of $f$ is defined,
  for all $u,v \in \wds$, by $u \lsim_f v$ if and only if the two following conditions
  are verified:
    $\dom(f)u^{-1} = \dom(f)v^{-1}$\label{itm:recDomBim}
    and $\sup_{w \in \dom(f)u^{-1}} \dist{f(wu)}{f(wv)} < \infty$\label{itm:sameOutBim}.

  The \emph{right syntactic congruence} $\rsim_f$ of $f$ is defined symmetrically,
  using a distance based on the longest common suffix.
\end{definition}


\begin{theorem}[\cite{ReuSch1991}]\label{thm:caracRatFct}
  Let $f \fct{\wds}{\wds[\alpp]}$ be a partial function.
  The following conditions are equivalent:
  \begin{enumerate}[(1)]
    \item $f$ is a rational transduction
    \item $\lsim_f$ has a finite index and $f^{-1}(L) \in \rat{\wds}$ for all
          $L \in \rat{\wds[\alpp]}$.
    \item $\rsim_f$ has a finite index and $f^{-1}(L) \in \rat{\wds}$ for all
          $L \in \rat{\wds[\alpp]}$.
  \end{enumerate}
\end{theorem}


\begin{definition}\label{def:leftRightCanAut}
  We define the \emph{canonical right automaton} $\Rcal_f$ of a rational function
  $f$ as the codeterministic automaton $\DFACong{\lsim_f}$ obtained
  from the left syntactic congruence $\lsim_f$ with
  $\setWCnd{\eqCl[\rsim_f]{w}}{w \in \dom(f)}$ as its set of initial states.
  Its \emph{canonical left automaton} $\Lcal_f$ is defined symmetrically using $\rsim_f$.
\end{definition}

\begin{definition}
  A bimachine $\Bcal$, with left and right automata $\Lcal$ and $\Rcal$ respectively,
  realizing a function $f$ is called \emph{minimal} if no other
  bimachine $\Bcal'$, with left and right automata $\Lcal'$ and $\Rcal'$ respectively,
  realizing $f$ is such that $\Lcal$ is finer than $\Lcal'$ and $\Rcal$ is finer than $\Rcal'$.
\end{definition}

For a total function $f$,~\cite{ReuSch1991} shows the existence of a canonical 
minimal bimachine, with $\Rcal_f$ as its right automaton.
Moreover, $\Rcal_f$ is minimal among all the right automata of the bimachines 
realizing $f$.
The same is true for $\Lcal_f$, which allows to define a left canonical bimachine
with a minimal left automaton.
This result was later extended to all rational functions in~\cite{FilGauLho2019}
where the canonical bimachines were shown to be computable in polynomial time.
The paper also provides \ptime algorithms to obtain the minimal right automaton
associated with a given left automaton and vice-versa, allowing the efficient 
minimization of any given bimachine.
We sum up the results relevant to this paper as follows:

\begin{theorem}[\cite{ReuSch1991,FilGauLho2019}]
  Bimachines can be minimized in \ptime.
  Moreover, given a rational function $f$, there exist a computable finite family 
  of minimal bimachines realizing $f$.
\end{theorem}

\begin{remark}\label{rmk:domains}
  Unlike Definition~\ref{def:bim} (which is the same as in~\cite{ReuSch1991})~\cite{FilGauLho2019}
  requires from both automata of a bimachine to recognize the domain of the function.
  This is important for minimization problems of the next sections
  (see Appendix~\ref{apx:bimachines-min}).
\end{remark}

\section{Relating Streaming String Transducers and bimachines}%
\label{sec:bimachine}

We start by introducing a subclass of \SST that will allow us to
obtain an equivalence with bimachines.

Let $\alp$ and $\alpp$ be two finite alphabets and let
$\Scal = (Q, \varSet, q_0, v_0, \sttTrSST, \regUpSST, \outFctSST)$
be an \SST from $\wds$ to $\wds[\alpp]$.
By swapping the contents of the registers if necessary, we may assume without
loss of generality that $\Scal$ has a fixed output register.
Meaning that there exists a register $\regOut \in \varSet$ such that, for all $q \in
  \dom(\outFctSST)$, $\outFctSST(q) = \regOut w$ for some $w \in \wds[\alpp]$.
Using the terminology of~\cite{FilKriTri2014}, given a register update
$X \coloneqq Y w$ of $\Scal$, the register $Y$ is said to \emph{flow} to $X$.
We will say that $\Scal$ has \emph{independent flows} if all the flows
of the register updates of its transitions only depend on the letter read rather
than both the letter and the state of its underlying automaton, or, more formally:
\begin{definition}
    An \SST $\Scal = (Q, \varSet, q_0, v_0, \sttTrSST, \regUpSST, \outFctSST)$ is
  said to have independent flows if for all $X \in \varSet$ and $\ltr \in \alp$,
  there exists $Y \in \varSet$ such that for all $q \in Q$, $\regUpSST(q,\ltr)(X) 
  = Yw$ for some word\footnote{Which can depend on the state $q$.} 
  $w \in \wds[\alpp]$.
\end{definition}

\begin{example}
  The \SST depicted on Figure~\ref{fig:SST_indep} has independent flows while the 
  equivalent \SST of Figure~\ref{fig:SST_dep} has a fixed output register but 
  flows depending on its states.
\end{example}

\begin{wrapfigure}{r}{0.55\textwidth}
	\vspace{-.4cm}
    \begin{center}
      \tikzset{parallel/.style={auto=right,->,
      to path={ let \p1=(\tikztostart),\p2=(\tikztotarget),
          \n1={atan2(\y2-\y1,\x2-\x1)},\n2={\n1+180}
          in ($(\tikztostart.{\n1})!1mm!270:(\tikztotarget.{\n2})$) --
          ($(\tikztotarget.{\n2})!1mm!90:(\tikztostart.{\n1})$) \tikztonodes}}}

 \scalebox{.9}{
\begin{tikzpicture}[auto, node distance=3cm]
  \node[state,
    initial, initial text = $\trn{}{X \coloneqq \e\\ Y \coloneqq \e}$,
    accepting below, accepting text = $X$
  ] (0) {};
  \node[state] (1) [right=of 0, accepting below, accepting text = $X$] {};

  \path[->]
    (0) edge[parallel] node[align=left] {$\trn{b}{X \coloneqq Yb\\ Y \coloneqq Xa}$} (1)
    (1) edge[parallel] node[align=left] {$\trn{a}{X \coloneqq Ya\\ Y \coloneqq Xb}$} (0)
    (0) edge[loop above] node[align=left] {$\trn{a}{X \coloneqq Xa\\ Y \coloneqq Yb}$} (0)
    (1) edge[loop above] node[align=left] {$\trn{b}{X \coloneqq Xb\\ Y \coloneqq Ya}$} (1)
  ;
\end{tikzpicture}
 }
      \caption{An \SST with a fixed output register and dependant flows.}
      \label{fig:SST_dep}
    \end{center}
	\vspace{-.8cm}
\end{wrapfigure}
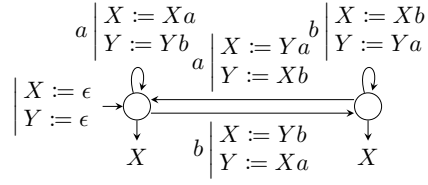
Barring some considerations concerning the domains, this restriction, in conjunction
with the fixed output register assumption, 
allows us to show a one-to-one correspondence between the \SST verifying these
conditions and the bimachines realizing the same transduction.

In the following, we consider the class of \SST that have independent flows, and 
a fixed output register. We denote it by $\SST_{\textsf{iffo}}$.

\begin{proposition}%
\label{prop:SST-bimachines}
Let $f \fct{\wds}{\wds[\alpp]}$ be a total function, and $k,n$ be two integers.

    There exists an $\SST_{\textsf{iffo}}$ that realizes $f$ with $n$ states and $k$ registers
    iff there exists a bimachine that realizes $f$ whose left automaton has $n$ states, and whose right automaton has $k$ states.
\end{proposition}

The idea is that the left automaton, corresponds to the underlying 
automaton of the \SST, while the right automaton specifies the flows of the registers
(See Appendix~\ref{apx:flow_indep} for details).
The condition of independence of the flows is required since the transitions of 
the right automaton of a bimachine do not depend on the states reached in its left one.
As a consequence,
we obtain:
\begin{corollary}\label{cor:regCpxSSTIndep}
Let $f$ be a total rational function.
  The register complexity of $f$, relative to the
  class $\SST_{\textsf{iffo}}$, is the number
  of states of the right automaton of a bimachine with a minimal right
  automaton realizing $f$. 
The state-register complexity of $f$, relative to the
  class $\SST_{\textsf{iffo}}$, is the following set:
$$\{(|\Lcal|,|\Rcal|) \mid \Bcal=(\Lcal,\Rcal,\lambda,\omega,\rho) \textup{ is a minimal bimachine realizing } f\}$$
\end{corollary}

Since bimachine minimization can be performed in \ptime, as we have seen in
Section~\ref{sec:algebra}, we can then solve the register minimization problem for
this class of \SST by computing a bimachine with a minimal right automaton
and converting it into an \SST using the construction above. Similarly,
in order to solve \PbSttRegMin, one can non-deterministically guess
a bimachine $\Bcal$ of the expected size, check that it is minimal, and that it recognizes $f$.
By definition, the size of $\Bcal$ is polynomial in $n$ and $k$, which we assume are given in unary (as was done in~\cite{DavReyTal2016})
and minimality and equivalence can be checked in \ptime. We obtain:

\begin{theorem}
  \PbRegMin (\resp  \PbSttRegMin) is decidable in \ptime (\resp \np) for $\SST_{\textsf{iffo}}$ realizing total functions.
\end{theorem}

Minimization results of~\cite{ReuSch1991,FilGauLho2019} require
from the right automaton of the bimachine to recognize the domain of the function,
which is a condition that is not required from an \SST since its states and register
updates are not necessarily correlated.
Thus, the characterization of the register complexity of Corollary~\ref{cor:regCpxSSTIndep}
cannot easily be extended to partial functions.
Moreover, even for total functions, the register complexity for the class
$\SST_{\textsf{iffo}}$ can be higher than the one for general \SST.
For instance, considering the function of Example~\ref{ex:FST_last},
one can verify that the register complexity of this 
  function relative to the class $\SST_{\textsf{iffo}}$ is three, while it can be a realized by an \SST with two registers (Figure~\ref{fig:SST_indep}).
  
%


\section{A new model of bimachines}%
\label{sec:asynchronous}


%
%
To drop the independent flow restriction, we will modify bimachines as follows: instead of
using both automata to read the input simultaneously in both directions,
what we will call an \emph{asynchronous bimachine} will first read the input
using its left automaton, 
and then, using its right automaton, it will read 
the run of the left automaton
and produce the output.
This is also reminiscent of the proof of~\cite[Theorem IV.5.2]{Ber1979} of Elgot and Mezei's decomposition.


\paragraph{Definition of the model}

  
  An asynchronous bimachine on two finite alphabets $\alp$ and $\alpp$ is a tuple
  $\Bcal = (\Lcal, \Rcal, \lOutFctBim, \outFctBim, \rOutFctBim)$ where
    its \emph{left automaton} is $\Lcal = (L, \lIni, \delta_\Lcal, F_\Lcal)$,
          a deterministic finite automaton on the alphabet $\alp$,
    its \emph{right automaton} is $\Rcal = (R, I_\Rcal, \delta_\Rcal, \rFin)$,
          a codeterministic finite automaton on the alphabet $L \times \alp$,
    and $\lOutFctBim \fct{I_\Rcal}{\wds[\alpp]}$, $\outFctBim
            \fct{R \times (L \times \alp) \times R}{\wds[\alpp]}$, and
          $\rOutFctBim \fct{F_\Lcal}{\wds[\alpp]}$ are respectively its
          \emph{left output}, \emph{output} and \emph{right output} functions.

A run of $\Bcal$ on a word $w = a_1 \dots a_n \in \wds$ is a finite sequence
$(l_j, r_j)_{j \in \intInterv{0}{n}}$ of elements of $L \times R$ such that
$l_0 = \lIni$, $(l_j)_{j \in \intInterv{0}{n}}$ is the run of $\Lcal$ on $w$,
$r_n = \rFin$ and $(r_j)_{j \in \intInterv{0}{n}}$
is the run of $\Rcal$ on the word $w_\Lcal = (l_0, a_1)\dots(l_{n-1}, a_n) \in \wds[
    (L \times \alp)]$.
The transduction $ \fctABim{\Bcal} \fct{\wds}{\wds[\alpp]}$ \emph{realized}
by $\Bcal$ is defined by
$$
   \fctABim{\Bcal}(w) = \lOutFctBim(r_0)\
  \prod_{j=1}^{n} \outFctBim\big( r_{j-1},\, (l_{j-1},a_j) \,, r_j \big)\
  \rOutFctBim(l_n)
$$


As described above, we can view $\fctABim{\Bcal}$ as the composition of 
a function realized by a sequential transducer $\Tcal_\Lcal$ annotating the input
by the states of $\Lcal$ and a cosequential transducer $\Tcal_\Rcal$ reading 
these annotated words.
We define these transducers more formally in Appendix~\ref{apx:async},
with an example in Appendix~\ref{apx:ex-async}.

\paragraph{Correspondence with streaming string transducers}
Using asynchronous bimachines we manage to characterize the full class of \SST:

\begin{theorem}\label{thm:caracSST}
  A rational function can be realized by an \SST with $\sCRA$ states and
  $\rCRA$ registers if and only if it can be realized by an asynchronous bimachine
  having a left automaton with $\sCRA$ states and a right automaton with
  $\rCRA$ states.
\end{theorem}

The conversions are the same as those of Proposition~\ref{prop:SST-bimachines}.
The equality of the domains and the removal of the restriction on the flows
of the \SST are direct consequences of the definition of asynchronous bimachines.
\begin{corollary}\label{cor:regCpxSST}
  The register complexity of a rational function $f$, relative to the class
  \SST, is the minimal number of states needed by the right
  automaton of an asynchronous bimachine to realize $f$.
\end{corollary}

\paragraph{Minimization of asynchronous bimachines}
By Corollary~\ref{cor:regCpxSST}, to solve \PbRegMin, we have to find an asynchronous
bimachine with a right automaton having the minimal number of states among all
the bimachines realizing the same function.
However, adapting the known algorithms of~\cite{ReuSch1991,FilGauLho2019} for
the minimization of bimachines to asynchronous bimachines seems rather challenging.

A first step would be to fix the left automaton and search for the smallest possible
right automaton that can be associated with it.
This corresponds to minimizing the number of registers of a given \SST without
changing its underlying automaton. We formalize this problem as follows:\\
\noindent
\textbf{Problem:} Fixed Underlying Automaton Register Minimization (\PbFARegMin)\\
\noindent \textit{Input:} an \SST $\Scal$ realizing a function $f$, and $k\in \mathbb{N}$\\
\noindent \textit{Question:} Does there exist an \SST $\Scal'$ that realizes $f$, has $k$ registers
and s.t. $\Scal$ and $\Scal'$ have the same underlying automaton?

\smallskip

Viewing an asynchronous bimachine $\Bcal$ as the composition of a sequential
transducer $\Tcal_\Lcal$ and a cosequential one $\Tcal_\Rcal$ as above, we want
to find a cosequential transducer $\Tcal$ with the minimal number of states such
that $\fctTrn{\Tcal}(w) = \fctTrn{\Tcal_\Rcal}(w)$ for all $w \in \Tcal_\Lcal(\wds)$.
Without loss of generality, we can consider $\fctTrn{\Tcal}$ to be a total function.
Furthermore, up to mirroring the input word, we can consider $\Tcal_\Rcal$ sequential
and search for a sequential transducer.

We have just reduced the problem of minimizing the right automaton of an asynchronous
bimachine, with a fixed left automaton, to the problem of extension of
sequential transducers, defined as:\\
\noindent
\textbf{Problem:} Extension of sequential transducers\\
\noindent \textit{Input:} a sequential transducer $\Tcal$ and
    $\rCRA \in \Nbb$\\
\noindent \textit{Question:} Does there exist a sequential transducer $\Tcal'$
with $\rCRA$ states such that 
$\fctTrn{\Tcal'}(w) = \fctTrn{\Tcal}(w)$ for all $w \in \dom(\fctTrn{\Tcal})$?

\smallskip

This problem has been studied in previous works~\cite{Pfl1973,ReuMer1986},
and shown to be \np-complete.
Regarding the problem \PbFARegMin, one can propose an \np procedure 
which follows a classical guess-and-check scheme. Adapting
the hardness proofs developed for the extension of sequential transducers, we prove:
\begin{proposition}%
\label{prop:NPhard}
Solving problem \PbFARegMin for an \SST with partial updates is \npc.
\end{proposition}

This shows that \PbFARegMin is hard, and we are thus interested
in identifying an algorithm to solve it using an automata-based approach.
In the particular case of letter-to-letter transducers, or, rational functions
with a prefix-closed domain, we show that the problem of extension of
a sequential transducer is equivalent to a refinement problem
for deterministic automata.

Given a \DFA $\Acal = (Q, q_0, \delta, F)$, we will call a precongruence
on $Q$ any reflexive and symmetric relation $\approx$ on $Q$ such that for all
$p,q \in Q$ and $\ltr \in \alp$, if $p \approx q$ then $p \cdot_\Acal \ltr \approx
  q \cdot_\Acal \ltr$.
Any precongruence $\approx$ gives rise to a relation on $\wds$, that we will also
denote by $\approx$, defined for all $u,v \in \wds$ by $u \approx v$ if and only if
$q_0 \cdot_\Acal u \approx q_0 \cdot_\Acal v$.
We can view $\approx$ as a compatibility relation on the states of $\Acal$,
which is not necessarily transitive, but ensures that any two runs of $\Acal$ on
the same word starting from two compatible states ends in two states that
are also compatible.

We want to obtain a \DFA with a smaller number of states by merging together the
states of $\Acal$ that are compatible.
More formally, we define:\\
\noindent
\textbf{Problem:} Minimal refinement Problem\\
\noindent \textit{Input:} a precongruence $\approx$ on the states of a \DFA and
    $\rCRA \in \Nbb$\\
\noindent \textit{Question:} Does there exist a \DFA finer than $\approx$ with at most
    $\rCRA$ states?

\begin{proposition}\label{prop:pbRefEqPbExt}
  The problem of extension of sequential transducers, for rational functions
  realizable by a letter-to-letter transducer and those with a prefix-closed
  domain, is equivalent to the minimal refinement problem.
\end{proposition}

We will now describe a construction, we call the \emph{subset expansion}, allowing
the enumeration of all the minimal \DFA refining a given precongruence.
It is similar to the classical subset construction for the determinization of \NFA,
but the subsets are formed with states that are pairwise compatible,
representing all the merges of states allowed by the precongruence.
We define it as follows.
\begin{definition}
  
  Let $\Acal = (Q_\Acal, q_0, \delta_\Acal, F_\Acal)$ be a \DFA and let $\approx$ be a
  precongruence on $Q$.
  We define the \NFA $\Ecal = (Q_\Ecal, I_\Ecal, \Delta_\Ecal, F_\Ecal)$ of the
  subset expansion of $\Acal$ relative to $\approx$ as follows:
    $Q_\Ecal = \setWCnd{P \subseteq Q_\Acal}{\forall p,q \in P,\, p \approx q}$,
    $I_\Ecal = \setWCnd{P \in Q_\Ecal}{q_0 \in P}$,
    $\Delta_\Ecal = \setWCnd{(P, \ltr, Q) \in Q_\Ecal \times \alp \times Q_\Ecal}
            {\forall p \in P,\, p \cdot_\Acal \ltr \in Q}$
    and $F_\Ecal = Q_\Ecal$.
\end{definition}

Given two automata $\Acal = (Q, I, \Delta, F)$ and $\Acal' = (Q', I', \Delta', F')$,
we call $\Acal'$ a \emph{subautomaton} of $\Acal$ if $Q' \subseteq Q$, $I' \subseteq I$,
$\Delta' \subseteq \Delta$ and $F' \subseteq F$.
Note that all the deterministic subautomata of $\Ecal$ are finer than $\approx$
by definition. We prove:

\begin{theorem}%
\label{thm:refinement}
  Let $\Acal$ be a \DFA, let $\approx$ be a precongruence on its states and
  let $\Ecal$ be the automaton of the subset expansion of $\Acal$ relative to
  $\approx$.

  If $\Bcal$ is a \DFA finer than $\approx$ then there exists
  a deterministic subautomaton $\Bcal'$ of $\Ecal$ such that $\congDFA{\Bcal}$
  is finer than $\congDFA{\Bcal'}$.
  Moreover, if $\Bcal$ is a minimal automaton among all the \DFA finer
  than $\approx$ then $\Bcal'$ is equal to $\Bcal$ up to state renaming.
\end{theorem}

We can then find all the solutions to the minimal refinement problem by enumerating
the deterministic subautomata of the automaton of the subset expansion having
a number of states bounded the given integer $\rCRA$.
Thanks to the equivalence of Proposition~\ref{prop:pbRefEqPbExt} and the correspondence
between \SST and asynchronous bimachines of Theorem~\ref{thm:caracSST}, this
allows us to find the minimal number of registers needed by an \SST with a fixed
underlying automaton having register updates of the form $X \coloneqq Y \ltr$ with
$\ltr \in \alp$ or such that the domain of the function it realizes is
suffix-closed~\footnote{Note the inversion of the reading direction since the right
  automaton we seek to minimize is actually codeterministic.}.

In the general case of the extension problem for arbitrary sequential transducers,
the possible delays between the outputs produced from each pairs of states prevents
the definition of a suitable compatibility relation and, to the best of our knowledge,
the problem remains open.

\section{On the tradeoff between states and registers}%
\label{sec:tradeoff}

First, observe that in an \SST, the domain is recognized by the underlying automaton,
which is deterministic. Let $f$ be a rational function. As a consequence, the \emph{state complexity} of
$f$ is at least the finite index of the Myhill Nerode
congruence associated with $\dom(f)$.
Moreover, thanks to our correspondence with bimachines, it is at most the index of the right syntactic congruence $\rsim_f$ of $f$.

Regarding the \emph{register complexity}, we recall the result of~\cite{DavReyTal2016}, with a precision on the 
number of states of the equivalent SST:
\begin{theorem}[\cite{DavReyTal2016}]
Let $f$ be a rational function realized by a functional transducer with $N$ states, and 
$k\in \mathbb{N}$ be an integer. One can decide whether there exists an \SST with 
$k$ registers that realizes $f$. In addition, when this is the case, then one can build such an \SST, and
its number of states is in $2^{O(N^k)}$.
\end{theorem}

This theorem shows that reducing the number of registers may induce an
\emph{exponential blowup} of the number of states. We show, using the 
next example, that this bound is tight, by providing an example
for which an exponential blowup of the number of states occurs.

\begin{example}
Let $\alp=\{a,b\}$, $N\in \mathbb{N}$, and consider the following function:
$f:u\sigma \in \alp^{N+1} \mapsto \sigma u$, where $\sigma\in \alp$ is a letter,
and $u\in \alp^N$ is a word on $\alp$ of length $N$.
This function can be realized by an \SST with two registers $X_a,X_b$ and $O(N)$ states. It simply copies the input word $u$ in both registers, but on the first transition, it adds an $a$ in $X_a$ and
a $b$ in $X_b$. On the last transition, it records the letter to decide which register to output.
It uses $O(N)$ states to check the length of the input word. 

This function can also be realized by an \SST with a single register, but one can show that this requires to store $u$ in the states, so as to produce the whole output on the last letter of the input,
yielding an exponential number of states.
\end{example}

\section{Conclusion}

Using the equivalence with bimachines, we have obtained strong 
and efficient minimization results for $\SST_{\textsf{iffo}}$ with total domain.
On the other side, the model
of asynchronous bimachines allows to capture the full class of \SST, but is more
difficult to minimize. 
%
%
This has led us to the study of the register minimization for a fixed underlying
automaton.
Next, the converse should be investigated.
\ie minimizing the left automaton with respect to a fixed right automaton
and checking the canonicity of the asynchronous bimachine obtained
by composing the two constructions.
As future work, another challenging research direction is the state-register minimization problem
for the full class of \SST.

%
%

\newpage
\bibliographystyle{splncs04}
\bibliography{biblio}

@Book{str94,
  author =	"Howard Straubing",
  title =	"Finite Automata, Formal Logic, and Circuit
		 Complexity",
  publisher =	"Birkh{\"a}user",
  address =	"Boston, Basel and Berlin",
  year = 	"1994",
}

@Article{Bu60,
  author = 	 {B\"uchi, J. R.},
  title = 	 {Weak second-order arithmetic and finite automata},
  journal = 	 {Zeitschrift f\"ur Mathematische Logik und Grundlagen der Mathematik},
  year = 	 {1960},
  OPTkey = 	 {},
  volume = 	 {6},
  number = 	 {1--6},
  pages = 	 {66-92},
  OPTmonth = 	 {},
  OPTnote = 	 {},
  OPTannote = 	 {}
}

@inproceedings{AluCer2010,
  author    = {Alur, Rajeev and Cerný, Pavol},
  booktitle = {{{IARCS}} Annual Conference on Foundations of Software Technology and Theoretical Computer Science, {{FSTTCS}} 2010, December 15-18, 2010, Chennai, India},
  year      = {2010},
  doi       = {10.4230/LIPICS.FSTTCS.2010.1},
  editor    = {Lodaya, Kamal and Mahajan, Meena},
  pages     = {1--12},
  publisher = {Schloss Dagstuhl - Leibniz-Zentrum für Informatik},
  series    = {{{LIPIcs}}},
  title     = {Expressiveness of Streaming String Transducers},
  volume    = {8}
}

@inproceedings{AluCer2011,
  author    = {Alur, Rajeev and Cerný, Pavol},
  booktitle = {Proceedings of the 38th {{ACM SIGPLAN-SIGACT}} Symposium on Principles of Programming Languages, {{POPL}} 2011, Austin, {{TX}}, {{USA}}, January 26-28, 2011},
  year      = {2011},
  doi       = {10.1145/1926385.1926454},
  editor    = {Ball, Thomas and Sagiv, Mooly},
  pages     = {599--610},
  publisher = {ACM},
  title     = {Streaming Transducers for Algorithmic Verification of Single-Pass List-Processing Programs}
}

@inproceedings{AluDADesRagYua2013,
  author    = {Alur, Rajeev and D'Antoni, Loris and Deshmukh, Jyotirmoy V. and Raghothaman, Mukund and Yuan, Yifei},
  booktitle = {28th Annual {{ACM}}/{{IEEE}} Symposium on Logic in Computer Science, {{LICS}} 2013, New Orleans, {{LA}}, {{USA}}, June 25-28, 2013},
  year      = {2013},
  doi       = {10.1109/LICS.2013.65},
  pages     = {13--22},
  publisher = {IEEE Computer Society},
  title     = {Regular Functions and Cost Register Automata}
}

@inproceedings{BenLhoRey2024a,
  author    = {Benalioua, Yahia Idriss and Lhote, Nathan and Reynier, Pierre{-}Alain},
  booktitle = {49th International Symposium on Mathematical Foundations of Computer Science, {{MFCS}} 2024, August 26-30, 2024, Bratislava, Slovakia},
  year      = {2024},
  doi       = {10.4230/LIPICS.MFCS.2024.23},
  editor    = {Kr{\'{a}}lovic, Rastislav and Kucera, Anton{\'{\i}}n},
  pages     = {23:1--23:15},
  publisher = {Schloss Dagstuhl - Leibniz-Zentrum für Informatik},
  series    = {{{LIPIcs}}},
  title     = {Minimizing Cost Register Automata over a Field},
  volume    = {306}
}

@article{
  BerstelB79,
  author = {Jean Berstel and Luc Boasson},
  title = {Transductions and context-free languages},
  journal = {Ed. Teubner},
  year = {1979},
  pages = {1--278}
}

@book{Ber1979,
  author    = {Berstel, Jean},
  year      = {1979},
  isbn      = {3-519-02340-7},
  publisher = {Teubner},
  series    = {Teubner Studienbücher : {{Informatik}}},
  title     = {Transductions and Context-Free Languages},
  url       = {https://www.worldcat.org/oclc/06364613},
  volume    = {38}
}

@inproceedings{DavReyTal2016,
  author    = {Daviaud, Laure and Reynier, Pierre-Alain and Talbot, Jean-Marc},
  booktitle = {Proceedings of the 31st Annual {{ACM}}/{{IEEE}} Symposium on Logic in Computer Science, {{LICS}} '16, New York, {{NY}}, {{USA}}, July 5-8, 2016},
  year      = {2016},
  doi       = {10.1145/2933575.2934549},
  editor    = {Grohe, Martin and Koskinen, Eric and Shankar, Natarajan},
  pages     = {857--866},
  publisher = {ACM},
  title     = {A Generalised Twinning Property for Minimisation of Cost Register Automata}
}

@book{Eil1974,
  author    = {Eilenberg, Samuel},
  year      = {1974},
  isbn      = {0-12-234001-9},
  publisher = {Academic Press},
  series    = {Pure and Applied Mathematics},
  title     = {Automata, Languages, and Machines. {{Volume A}}},
  url       = {https://www.worldcat.org/oclc/310535248},
  volume    = {A}
}

@Article{tra61short,
  author =	"Boris Avraamovich Trakhtenbrot",
  title =	"Finite automata and logic of monadic predicates (in
		 {R}ussian)",
  journal =	"Dokl.\ Akad.\ Nauk SSSR",
  year = 	"1961",
  volume =	"140",
  pages =	"326--329",
}

@Article{Elg61,
  author = 	 {Elgot, C. C.},
  title = 	 {Decision Problems of Finite Automata Design and Related Arithmetics},
  journal = 	 {In Transactions of the American Mathematical Society},
  year = 	 {1961},
  OPTkey = 	 {},
  volume = 	 {98},
  number = 	 {1},
  pages = 	 {21--51},
  OPTmonth = 	 {},
  OPTnote = 	 {},
  OPTannote = 	 {}
}

@article{FilGauLho2019,
  author       = {Filiot, Emmanuel and Gauwin, Olivier and Lhote, Nathan},
  year         = {2019},
  doi          = {10.23638/LMCS-15(4:16)2019},
  journal = {Logical Methods in Computer Science},
  number       = {4},
  shortjournal = {Log. Methods Comput. Sci.},
  title        = {Logical and Algebraic Characterizations of Rational Transductions},
  volume       = {15}
}

@inproceedings{FilKriTri2014,
  author    = {Filiot, Emmanuel and Krishna, Shankara Narayanan and Trivedi, Ashutosh},
  booktitle = {34th International Conference on Foundation of Software Technology and Theoretical Computer Science, {{FSTTCS}} 2014, December 15-17, 2014, New Delhi, India},
  year      = {2014},
  doi       = {10.4230/LIPICS.FSTTCS.2014.147},
  editor    = {Raman, Venkatesh and Suresh, S. P.},
  pages     = {147--159},
  publisher = {Schloss Dagstuhl - Leibniz-Zentrum für Informatik},
  series    = {{{LIPIcs}}},
  title     = {First-Order Definable String Transformations},
  volume    = {29}
}

@article{FilRey2016,
  author       = {Filiot, Emmanuel and Reynier, Pierre-Alain},
  year         = {2016},
  doi          = {10.1145/2984450.2984453},
  journal = {ACM SIGLOG News},
  number       = {3},
  pages        = {4--19},
  title        = {Transducers, Logic and Algebra for Functions of Finite Words},
  volume       = {3}
}

@article{FilRey2021,
  author       = {Filiot, Emmanuel and Reynier, Pierre-Alain},
  year         = {2021},
  doi          = {10.3233/FI-2021-1998},
  journal = {Fundam. Informaticae},
  number       = {1--2},
  pages        = {59--76},
  title        = {Copyful Streaming String Transducers},
  volume       = {178}
}

@article{Pfl1973,
  author       = {Pfleeger, Charles P.},
  year         = {1973},
  doi          = {10.1109/T-C.1973.223655},
  journal = {IEEE Trans. Computers},
  number       = {12},
  pages        = {1099--1102},
  title        = {State Reduction in Incompletely Specified Finite-State Machines},
  volume       = {22}
}

@article{ReuMer1986,
  author       = {Reusch, Bernd and Merzenich, Wolfgang},
  year         = {1986},
  doi          = {10.1007/BF00263650},
  journal = {Acta Informatica},
  number       = {6},
  pages        = {663--678},
  title        = {Minimal Coverings for Incompletely Specified Sequential Machines},
  volume       = {22}
}

@article{ReuSch1991,
  author       = {Reutenauer, Christophe and Schützenberger, Marcel Paul},
  year         = {1991},
  doi          = {10.1137/0220042},
  journal = {Siam Journal On Computing},
  number       = {4},
  pages        = {669--685},
  shortjournal = {SIAM J. Comput.},
  title        = {Minimization of Rational Word Functions},
  volume       = {20}
}

@article{Sch1961b,
  author       = {Schützenberger, Marcel Paul},
  year         = {1961},
  doi          = {10.1016/S0019-9958(61)80006-5},
  journal = {Inf. Control.},
  number       = {2--3},
  pages        = {185--196},
  title        = {A Remark on Finite Transducers},
  volume       = {4}
}

\newpage
\appendix

\section{Examples}
\subsection{Example of a bimachine}\label{apx:ex-bimachine}
\begin{figure}
  \centering
\subfloat[Left automaton $\Lcal$]{\begin{tikzpicture}[auto, initial text=]
  \node[state, accepting right, initial] (v0) {$\lIni$};
  \node[state, accepting right] (v1) [below right= of v0] {$l_b$};
  \node[state, accepting right] (v2) [above right= of v0] {$l_a$};
  \path[->] (v0) edge node[swap] {$b$} (v1);
  \path[->] (v0) edge node {$a$} (v2);
  \path[->] (v1) edge[loop below] node {$a, b$} ();
  \path[->] (v2) edge[loop above] node {$a, b$} ();
\end{tikzpicture}\label{fig:Bim_swap_left}}
\hfil
\subfloat[Right automaton $\Rcal$]{\begin{tikzpicture}[auto, initial text=]
  \node[state, accepting right, initial] (v0) {$\rFin$};
  \node[state, initial] (v1) [below left= of v0] {$r_b$};
  \node[state, initial] (v2) [above left= of v0] {$r_a$};
  \path[->] (v1) edge node[swap] {$b$} (v0);
  \path[->] (v1) edge[loop below] node {$a, b$} ();
  \path[->] (v2) edge node {$a$} (v0);
  \path[->] (v2) edge[loop above] node {$a, b$} ();
\end{tikzpicture}\label{fig:Bim_swap_right}}
    \caption{A bimachine}\label{fig:Bim_swap}
\end{figure}
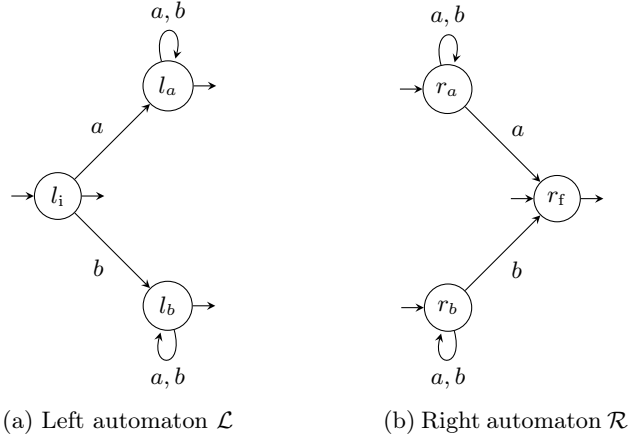

\begin{example}\label{ex:bimSwap}

  Let $\alp = \alpp = \set{a,b}$ and consider the bimachine
  $\Bcal = (\Lcal, \Rcal, \lOutFctBim, \outFctBim, \rOutFctBim)$
  whose automata are depicted on Figure~\ref{fig:Bim_swap} and whose output
  functions are defined, for all $l \in L$, $r \in R$ and $\ltr \in \alp$, by
  $\lOutFctBim(r) = \rOutFctBim(l) = \e$, $\outFctBim(\lIni, \ltr, r_\tau)
    = \outFctBim(l_\tau, \ltr, \rFin) = \tau$ for all $\tau \in \alp$ and
  $\outFctBim(l, \ltr, r) = \ltr$ in the other cases.

  It realizes a rational function swapping the first and the last letters of its
  input word.

  The left automaton remembers the first letter of the word until it reaches the
  end of the word and allows $\outFctBim$ to output it instead of the last letter
  while, symmetrically, the right automaton remembers the last letter and allows
  its output at the beginning.

  An example of a run of $\Bcal$ on the word $abab \in \wds$ is depicted on
  Figure~\ref{fig:Bim_swap_run}.

  \tikzFig{Bim_swap_run}{A run of a bimachine}
\end{example}

\subsection{Example of an asynchronous bimachine}\label{apx:ex-async}

\begin{example}
  The bimachine of Example~\ref{ex:bimSwap} viewed as an asynchronous bimachine
  is defined by $\Bcal = (\Lcal, \Rcal, \lOutFctBim, \outFctBim, \rOutFctBim)$
  where $\Lcal$ is the \DFA of Figure~\ref{fig:Bim_swap_left},
  $\Rcal = (R, I_\Rcal, \delta_\Rcal, \rFin)$ is defined by
  $R = I_\Rcal = \set{\rFin, r_a, r_b}$ and, for all $\ltr, \tau \in \alp$ and
  $l \in L$, $(l,\ltr) \cdot_\Rcal \rFin = r_\ltr$ and $(l,\ltr) \cdot_\Rcal r_\tau = r_\tau$.

  Its outputs are defined, for all $l \in L$, $r \in R$ and $\ltr \in \alp$,
  by $\lOutFctBim(r) = \rOutFctBim(l) = \e$, $\outFctBim\big(r_\tau, (\lIni, \ltr), r_\tau\big)
    = \outFctBim\big(r_\ltr, (l_\tau, \ltr), \rFin\big) = \tau$ for all
  $\tau \in \alp$ and $\outFctBim(r', (l, \ltr), r) = \ltr$ in the other cases,
  where $r' = (l, \ltr)\cdot_\Rcal r$.

  An example of a run of $\Bcal$ on the same word as in Figure~\ref{fig:Bim_swap_run}
  is depicted on Figure~\ref{fig:Bim_swap_run_async}.

  \tikzFig{Bim_swap_run_async}{}
\end{example}

\section{Precisions on minimality}\label{apx:bimachines-min}

As noted in Remark~\ref{rmk:domains}, there is a subtle difference between the 
definition of bimachines of~\cite{ReuSch1991} (which we use here) and the one of~\cite{FilGauLho2019}
regarding the domains of partial functions.
The assumptions of~\cite{FilGauLho2019} can be made without affecting the 
expressiveness of the model and leads to useful properties such as the existence 
of a finite family of minimal bimachines.
However, we avoid it here since forcing the right automaton to recognize the
domain of the function can lead to an unnecessary blowup in the number of registers
of the corresponding \SST.

The minimality of the canonical automata can be restated in our setting as follows:
\begin{proposition}\label{prop:canAutMinBim}
  Let $\Bcal = (\Lcal, \Rcal, \lOutFctBim, \outFctBim, \rOutFctBim)$ be a bimachine
  and let $f = \fctBim{\Bcal}$.

  If $\lanNFA{\Rcal} = \dom(f)$ then $\Rcal$ is finer than $\lsim_f$
  and, symmetrically,
  if $\lanNFA{\Lcal} = \dom(f)$ then $\Lcal$ is finer than $\rsim_f$.
\end{proposition}

In other words, Proposition~\ref{prop:canAutMinBim} states that the canonical
right (\resp left) automaton of a rational transduction $f$ is minimal among
all the right (\resp left) automata of the bimachines realizing $f$ with a
right (\resp left) automaton recognizing $\dom(f)$.

\begin{remark}
  Note that, in the case of total functions, both definitions of bimachines coincide,
  which ensures the correctness of the results of Section~\ref{sec:bimachine}.
\end{remark}

\section{Proofs of Section~\ref{sec:bimachine}}\label{apx:flow_indep}

\subsection{From streaming string transducers to bimachines}
\begin{proposition}\label{prop:SSTIndepToBim}
  Given an \SST with $\sCRA$ states and $\rCRA$ registers having independent flows
  and a fixed output register, we can define a bimachine realizing the same function
  having a left automaton with $\sCRA$ states and a right automaton with $\rCRA$ states.
\end{proposition}

The idea is that the assumptions on $\Scal$ allow us to define an automaton
describing its flows: its set of states is the set of registers of $\Scal$ and,
for all $\ltr \in \alp$, it has a transition from each register to the register
it flows to on the register updates of $\Scal$ corresponding to the letter $\ltr$.
This automaton happens to be codeterministic, choosing the fixed output register
of $\Scal$ as its final state, and can be used as the right automaton of a
bimachine together with the underlying automaton of $\Scal$ as its left
automaton.

\begin{proof}
  Let $\Scal = (Q, \varSet, q_0, v_0, \sttTrSST, \regUpSST, \outFctSST)$ be an
  \SST with independent flows and a fixed output register $\regOut \in \varSet$.

  We define a bimachine $\Bcal = (\Lcal, \Rcal, \lOutFctBim, \outFctBim, \rOutFctBim)$
  equivalent to $\Scal$ as follows:
  \begin{itemize}
    \item $\Lcal = (Q, q_0, \delta_\Lcal, \dom(\outFctSST))$ where, for all
          $q \in Q$ and $\ltr \in \alp$, $q \cdot_\Lcal \ltr = \sttTrSST(q,\ltr)$
          is the underlying automaton of $\Scal$.
    \item $\Rcal = (\varSet, \varSet, \delta_\Rcal, \regOut)$ where, for all
          $X \in \varSet$ and $\ltr \in \alp$ with $Y \in \varSet$ such that
          $\regUpSST (q,\ltr)(X) = Yu$ for some $q \in Q$ and $u \in \wds[\alpp]$,
          $\ltr \cdot_\Rcal X = Y$.
    \item For all $X \in \varSet$, we define $\lOutFctBim(X) = v_0(X)$.
    \item For all $q \in Q$, $\ltr \in \alp$ and $X \in \varSet$ with $\regUpSST
            (q,\ltr)(X) = Yu$ for some $u \in \wds[\alpp]$ (by definition $Y =
            \ltr \cdot_\Rcal X$), we define $\outFctBim(q,\ltr,X) = u$.
    \item For all $q \in \dom(\outFctSST)$ with $\outFctSST(q) = \regOut u$
          for some $u \in \wds[\alpp]$, we define $\rOutFctBim(X) = u$.
  \end{itemize}

  It is straightforward to check that $\dom(\fctBim{\Bcal}) = \lanNFA{\Lcal} =
    \dom(\fctSST{\Scal})$.
  Moreover, we can show by induction that, for all $X \in \varSet$ and for all
  $w \in \dom(\fctSST{\Scal})$ where $(q_i,v_i)_{i \in \intInterv{0}{n}}$ is the
  unique execution of $\Scal$ on $w$,
  \[
    v_n(X) = \lOutFctBim(w \cdot_\Rcal X)\, \outFctBim(q_0, w, X)
  \]

  Indeed, for $w = \e$, $v_0(X) = \lOutFctBim(X) = \lOutFctBim(\e \cdot_\Rcal X)\,
    \outFctBim(q_0, \e, X)$ by definition, and assuming the equality holds for
  all words of length $n \in \Nbb$, let $w \in \wds$ such that $\lenWd{w}
    = n$, let $\ltr \in \alp$ and let $(q_i,v_i)_{i \in \intInterv{0}{n+1}}$ be the
  unique execution of $\Scal$ on $w\ltr$ with $Y \in \varSet$ and $u \in \wds[\alpp]$
  such that $\regUpSST(q_n, \ltr)(X) = Yu$.
  By definition, $q_n = q_0 \cdot_\Lcal w$, $Y = \ltr \cdot_\Rcal X$ and $u =
    \outFctBim(q_n, \ltr, X)$.
  Thus,
  \begin{align*}
    v_{n+1}(X) & = v_n(Y)\, u                                                       \\
               & = \lOutFctBim(w \cdot_\Rcal Y)\, \outFctBim(q_0, w, Y)\,
    \outFctBim(q_n, \ltr, X)                                                        \\
               & = \lOutFctBim(w \ltr \cdot_\Rcal X)\,
    \outFctBim(q_0, w, \ltr \cdot_\Rcal X)\, \outFctBim(q_0 \cdot_\Lcal w, \ltr, X) \\
    v_{n+1}(X) & = \lOutFctBim(w\ltr \cdot_\Rcal X)\, \outFctBim(q_0, w\ltr, X)
  \end{align*}

  Let $\outFctSST(q_n)= \regOut u$ with $u \in \wds[\alpp]$.
  By definition, $q_n = q_0 \cdot_\Lcal w$ and $u = \rOutFctBim(\lIni \cdot_\Lcal w)$.
  Thus, the equivalence of $\Scal$ and $\Bcal$ follows immediately
  \begin{align*}
    \fctSST{\Scal}(w) & = v_n (\regOut)\, u                       \\
                      & = \lOutFctBim(w \cdot_\Rcal \regOut)\,
    \outFctBim (q_0, w, \regOut)\, \rOutFctBim(q_0 \cdot_\Lcal w) \\
    \fctSST{\Scal}(w) & = \fctBim{\Bcal}(w)
  \end{align*}
\end{proof}

\subsection{From bimachines to streaming string transducers}
For the converse, we can essentially reverse the construction of the proof of
Proposition~\ref{prop:SSTIndepToBim}.
However, there is a slight break of symmetry since bimachines are allowed to
recognize the domain of the function they realize either using their left or
right automaton but \SST can only recognize it using their underlying automaton.
Thus, we obtain the following:

\begin{proposition}\label{prop:BimToSSTIndep}
  Given a bimachine $\Bcal$ having a left automaton with $\sCRA$ states and a
  right automaton with $\rCRA$ states, we can define an \SST $\Scal$ with
  $\sCRA$ states and $\rCRA$ registers such that $\dom(\fctBim{\Bcal}) \subseteq
    \dom(\fctSST{\Scal})$ and $\fctSST{\Scal}(w) = \fctBim{\Bcal}(w)$
  for all $w \in \dom(\fctBim{\Bcal})$.

  In particular, $\Scal$ is equivalent to $\Bcal$ when the realized transduction
  is a total function.

  Moreover, $\Scal$ has independent flows and a fixed output register.
\end{proposition}

\begin{proof}
  Let $\Bcal = (\Lcal, \Rcal, \lOutFctBim, \outFctBim, \rOutFctBim)$ be a bimachine
  with $\Lcal = (L, \lIni, \delta_\Lcal, F_\Lcal)$ and
  $\Rcal = (R, I_\Rcal, \delta_\Rcal, \rFin)$.
  We define an \SST
  $\Scal = (Q, \varSet, q_0, v_0, \sttTrSST, \regUpSST, \outFctSST)$ as follows.

  The underlying automaton of $\Scal$ is $\Lcal$.
  \ie $Q = L$, $q_0 = \lIni$, $\dom(\outFctSST) = F_\Lcal$ and, for all
  $l \in L$ and $\ltr \in \alp$, $\sttTrSST(l,\ltr) = l \cdot_\Lcal \ltr$.

  For the registers and their updates, $\varSet = R$ and
  \begin{itemize}
    \item For all $r \in R$, $v_0$ is defined by $v_0(r) = \lOutFctBim(r)$ if
          $r \in \dom(\lOutFctBim)$ and $v_0(r) = \e$ otherwise\footnote{This
            choice is arbitrary as the initial value of these registers will be
            irrelevant during the execution of S on words of $\dom(\fctBim{\Bcal})$
            but the definition of \SST requires from all the registers to always
            have a value.
            This is the first source of the discrepancy between the domains of
            $\fctSST{\Scal}$ and $\fctBim{\Bcal}$.}.
    \item For all $l \in L$, $\ltr \in \alp$ and $r \in R$, $\regUpSST$ is defined
          by $\regUpSST(l,\ltr)(r) = r'u$ where $r' = \ltr \cdot_\Rcal r$ and
          $u = \outFctBim (l, \ltr, r)$ if $l \cdot_\Lcal \ltr$ and
          $\ltr \cdot_\Rcal r$ are defined and $\regUpSST(l,\ltr)(r)= r$
          otherwise\footnote{This choice is also arbitrary
            but required by the definition of \SST (see the previous footnote).
            It is the second source of the discrepancy between the domains of
            $\fctSST{\Scal}$ and $\fctBim{\Bcal}$.}.
    \item For all $l \in F_\Lcal$, $\outFctSST(l) = \rFin\, \rOutFctBim(l)$.
  \end{itemize}

  Note that $\Scal$ has a fixed output register $\rFin$ and independent flows since
  they only depend on the transitions of $\Rcal$ and not on $\Lcal$.

  We can easily check that $\dom(\fctBim{\Bcal}) = \lanNFA{\Lcal} \cap \lanNFA{\Rcal}
    \subseteq \lanNFA{\Lcal} = \dom(\fctSST{\Scal})$ and prove by induction, as in
  the proof of Proposition~\ref{prop:SSTIndepToBim}, that for all $r \in R$ and
  $w \in \dom(\fctBim{\Bcal})$ where $(q_i,v_i)_{i \in \intInterv{0}{n}}$ is the
  unique execution of $\Scal$ on $w$,
  \[
    v_n(r) = \lOutFctBim(w \cdot_\Rcal r)\, \outFctBim(q_0, w, r)
  \]

  Thus, from the same arguments, $\fctSST{\Scal}(w) = \fctBim{\Bcal}(w)$ for all
  $w \in \dom(\fctBim{\Bcal})$.
\end{proof}

\section{Complements on asynchronous bimachines}\label{apx:async}

Let $\Bcal = (\Lcal, \Rcal, \lOutFctBim, \outFctBim, \rOutFctBim)$ be an asynchronous
bimachine as in Section~\ref{sec:asynchronous}.
Let us give a formal definition of the sequential transducer 
$\Tcal_\Lcal = (L, \iniOFct_\Lcal, \oFct_\Lcal, \finOFct_\Lcal)$
and the cosequential transducer 
$\Tcal_\Rcal = (R, \iniOFct_\Rcal, \oFct_\Rcal, \finOFct_\Rcal)$ such that 
$\fctABim{\Bcal} = \fctTrn{\Tcal_\Rcal} \circ \fctTrn{\Tcal_\Lcal}$.

Let $\alp' = \alp \cup \set{\lhd}$, the additional symbol $\lhd$ will serve as an
\emph{end of word} marker\footnote{It is merely a technical detail to allow
  $\Tcal_\Rcal$ to take $\rOutFctBim$ into account while remaining cosequential.
  It can safely be ignored when $\rOutFctBim(q) = \e$ for all $q \in F_\Lcal$.
}.

$\Tcal_\Lcal$ is sequential and realizes a function $\fctTrn{\Tcal_\Lcal} \fct{\wds}
  {\wds[(L \times \alp')]}$ labeling the input by the states of $\Lcal$.
It is defined by
\begin{itemize}
  \item $\iniOFct_\Lcal (\lIni) = \e$
  \item $\oFct_\Lcal (l, \ltr, l \cdot_\Lcal \ltr) = (l, \ltr)$, for all $l \in L$
        and $\ltr \in \alp$ such that $l \cdot_\Lcal \ltr$ is defined
  \item $\finOFct_\Lcal(l) = (l, \lhd)$ for all $l \in F_\Lcal$
\end{itemize}

$\Tcal_\Rcal$ is cosequential and reads the labeled words to realize the transduction.
It is defined by
\begin{itemize}
  \item $\iniOFct_\Rcal (r) = \lOutFctBim(r)$ for all $r \in I_\Rcal$
  \item $\oFct_\Rcal \big(r',\, (l, \ltr),\, r \big) =
          \outFctBim \big(r',\, (l, \ltr),\, r \big)$, where $r' = (l, \ltr) \cdot_\Rcal r$
        for all $l \in L$, $\ltr \in \alp$ and $r \in R$ such that $l \cdot_\Lcal \ltr$
        and $(l, \ltr) \cdot_\Rcal r$ are defined, and
        $\oFct_\Rcal \big(\rFin,\, (l, \lhd),\, \rFin \big) = \rOutFctBim (l)$
        for all $l \in F_\Lcal$.
  \item $\finOFct_\Rcal (\rFin) = \e$
\end{itemize}

A simple induction on the length of the input word shows that $\fctABim{\Bcal}
  = \fctTrn{\Tcal_\Rcal} \circ \fctTrn{\Tcal_\Lcal}$

Without loss of generality, and for the purposes of this section\footnote{
  \ie to obtain the equality of the domains in the correspondence we will show
  with \SST.
}, we will always consider all the states of $\Rcal$
to be initial (\ie $I_\Rcal = R$) and require that, for all $l \in L$, $r \in R$
and $\ltr \in \alp$, a state $(l, \ltr) \cdot_\Rcal r$ of $\Rcal$ exists
whenever a state $l \cdot_\Lcal \ltr$ of $\Lcal$ exists.
Consequently, $\Lcal$ will be tasked with recognizing the domain of $\fctABim{\Bcal}$
since these requirements ensure that if $w \in \lanNFA{\Lcal}$ then
$w_\Lcal \in \lanNFA{\Rcal}$.

Equivalently, viewing $\fctABim{\Bcal}$ as the composition of a function realized
by a sequential transducer $\Tcal_\Lcal$ and a cosequential transducer $\Tcal_\Rcal$
as described above, this assumption can be restated as requiring
$\fctTrn{\Tcal_\Lcal} (\wds) \subseteq \dom (\fctTrn{\Tcal_\Rcal})$.

\subsection{Proof of Theorem~\ref{thm:caracSST}} 

\begin{proof}
  Given an \SST $\Scal = (Q, \varSet, q_0, v_0, \sttTrSST, \regUpSST, \outFctSST)$
  with a fixed output register\footnote{Remember that this assumption can be made
    without loss of generality if no conditions on the flows are imposed.}
  $\regOut \in \varSet$ define an asynchronous bimachine
  $\Bcal = (\Lcal, \Rcal, \lOutFctBim, \outFctBim, \rOutFctBim)$ as follows:
  \begin{itemize}
    \item $\Lcal = (Q, q_0, \delta_\Lcal, \dom(\outFctSST))$ where, for all
          $q \in Q$ and $\ltr \in \alp$, $q \cdot_\Lcal \ltr = \sttTrSST(q,\ltr)$
          is the underlying automaton of $\Scal$.
    \item $\Rcal = (\varSet, \varSet, \delta_\Rcal, \regOut)$ where, for all
          $X \in \varSet$ and $(q, \ltr) \in Q \times \alp$ with $Y \in \varSet$
          such that $\regUpSST (q,\ltr)(X) = Yu$ for some $u \in \wds[\alpp]$,
          $(q, \ltr) \cdot_\Rcal X = Y$.
    \item For all $X \in \varSet$, we define $\lOutFctBim(X) = v_0(X)$.
    \item For all $(q, \ltr) \in Q \times \alp$ and $X \in \varSet$ with $\regUpSST
            (q,\ltr)(X) = Yu$ for some $u \in \wds[\alpp]$ (by definition $Y =
            (q, \ltr) \cdot_\Rcal X$), we define $\outFctBim\big(Y,(q, \ltr),X\big)
            = u$.
    \item For all $q \in \dom(\outFctSST)$ with $\outFctSST(q) = \regOut u$
          for some $u \in \wds[\alpp]$, we define $\rOutFctBim(X) = u$.
  \end{itemize}

  It is straightforward to check that $\dom(\fctBim{\Bcal}) = \lanNFA{\Lcal} =
    \dom(\fctSST{\Scal})$.

  For all $w = a_1 \dots a_n \in \dom(\fctSST{\Scal})$,
  let $(q_i,v_i)_{i \in \intInterv{0}{n}}$ be the unique execution of $\Scal$ on
  $w$.
  Note that $(q_i)_{i \in \intInterv{0}{n}}$ is also, by definition, the unique
  execution of $\Lcal$ on $w$.
  For all $X \in \varSet$ and for all $i \in \intInterv{1}{n}$, let $X_n = X$ and
  $X_{i-1} =  (q_{i-1}, a_i) \cdot_\Rcal X_i$.
  We can show by induction on the length of $w$ that
  \[
    v_n(X) = \lOutFctBim(X_0) \prod_{i=1}^{n} \outFctBim(X_{i-1}, (q_{i-1}, a_i), X_i)
  \]

  Indeed, for $w = \e$, $v_0(X) = \lOutFctBim(X)$ by definition, and assuming
  the equality holds for all words of length $n \in \Nbb$, let $w = a_1 \dots a_n
    \in \wds$, let $a_{n+1} \in \alp$, let $(q_i,v_i)_{i \in \intInterv{0}{n+1}}$
  be the unique execution of $\Scal$ on $wa_{n+1}$ with $Y \in \varSet$ and
  $u \in \wds[\alpp]$ such that $\regUpSST(q_n, a_{n+1})(X) = Yu$ and
  define the sequence $(X_i)_{i \in \intInterv{0}{n+1}}$ as above with $X_{n+1} = X$
  and $X_{i-1} =  (q_{i-1}, a_i) \cdot_\Rcal X_i$ for all $i \in \intInterv{1}{n+1}$.
  Note that, by definition, $X_n = Y$ and $u = \outFctBim(X_n, (q_n, a_{n+1}), X_{n+1})$.
  Thus,
  \begin{align*}
    v_{n+1}(X) & = v_n(Y)\, u                            \\
               & = \Big(\lOutFctBim(X_0) \prod_{i=1}^{n}
    \outFctBim(X_{i-1}, (q_{i-1}, a_i), X_i)\Big)\,
    \outFctBim(X_n, (q_n, a_{n+1}), X_{n+1})             \\
    v_{n+1}(X) & = \lOutFctBim(X_0) \prod_{i=1}^{n+1}
    \outFctBim(X_{i-1}, (q_{i-1}, a_i), X_i)
  \end{align*}

  Let $\outFctSST(q_n)= \regOut u$ with $u \in \wds[\alpp]$.
  By definition, $u = \rOutFctBim(q_n)$ and the sequence $(X_i)_{i \in \intInterv{0}{n}}$
  with $X_n = \regOut$ and $X_{i-1} =  (q_{i-1}, a_i) \cdot_\Rcal X_i$ for all
  $i \in \intInterv{1}{n}$ is the unique execution of $\Rcal$ on $w$.
  Thus, the equivalence of $\Scal$ and $\Bcal$ follows immediately
  \begin{align*}
    \fctSST{\Scal}(w) & = v_n (\regOut)\, u \\
                      & = \lOutFctBim(X_0)\
    \prod_{i=1}^{n} \outFctBim\big( X_{i-1},\, (q_{i-1},a_i) \,, X_i \big)\
    \rOutFctBim(q_n)                        \\
    \fctSST{\Scal}(w) & = \fctBim{\Bcal}(w)
  \end{align*}

  Conversely, given an asynchronous bimachine
  $\Bcal = (\Lcal, \Rcal, \lOutFctBim, \outFctBim, \rOutFctBim)$
  with $\Lcal = (L, \lIni, \delta_\Lcal, F_\Lcal)$ and
  $\Rcal = (R, I_\Rcal, \delta_\Rcal, \rFin)$,
  let us define
  $\Scal = (Q, \varSet, q_0, v_0, \sttTrSST, \regUpSST, \outFctSST)$,
  an equivalent \SST, as follows.

  The underlying automaton of $\Scal$ is $\Lcal$.
  \ie $Q = L$, $q_0 = \lIni$, $\dom(\outFctSST) = F_\Lcal$ and, for all
  $l \in L$ and $\ltr \in \alp$, $\sttTrSST(l,\ltr) = l \cdot_\Lcal \ltr$.

  For the registers and their updates, $\varSet = R$ and
  \begin{itemize}
    \item For all $r \in R$, $v_0$ is defined by $v_0(r) = \lOutFctBim(r)$.
    \item For all $l \in L$, $\ltr \in \alp$ and $r \in R$, $\regUpSST$ is defined
          by $\regUpSST(l,\ltr)(r) = r'u$ where $r' = (l, \ltr) \cdot_\Rcal r$ and
          $u = \outFctBim (r', (l, \ltr), r)$.
    \item For all $l \in F_\Lcal$, $\outFctSST(l) = \rFin\, \rOutFctBim(l)$.
  \end{itemize}

  By the assumptions we made on asynchronous bimachines,
  $\dom(\fctBim{\Bcal}) = \lanNFA{\Lcal} = \dom(\fctSST{\Scal})$
  and we can prove by induction, as in the previous case, that for all $r \in R$
  and $w = a_1 \dots a_n \in \dom(\fctBim{\Bcal})$ where
  $(q_i,v_i)_{i \in \intInterv{0}{n}}$ is the unique execution of $\Scal$ on $w$
  and the sequence $(r_i)_{i \in \intInterv{0}{n}}$ is defined by $r_n = r$
  and $r_{i-1} = (q_{i-1}, a_i) \cdot_\Rcal r_i$ for all $i \in \intInterv{1}{n}$,
  \[
    v_n(r) = \lOutFctBim(r_0) \prod_{i=1}^{n} \outFctBim(r_{i-1}, (q_{i-1}, a_i), r_i)
  \]

  Thus, from the same arguments, $\fctSST{\Scal} = \fctBim{\Bcal}$.
\end{proof}

\subsection{Proof of Proposition~\ref{prop:NPhard}}
\begin{proof}
In this proof, we aim to show that, given as input an \SST $\Scal$ with
partial updates, and an integer $k$, determining whether there exists
an equivalent \SST $\Scal'$ (with total updates) with $k$ registers,
whose underlying automaton is the same as the one of $\Scal$ is \np-complete.

Let $\Scal = (Q, \varSet, q_0, v_0, \sttTrSST, \regUpSST, \outFctSST)$
be an \SST with partial updates. We denote by $\Acal$
its underlying \DFA.

We thus have to show that the problem is in \np and is \np-hard.

\paragraph{Membership in \np}
For membership in \np, we guess the \SST $\Scal'$, which
is of polynomial size. Then we have to check the equivalence.
This requires checking the equivalence of the domains, and of the
functions realized. Assuming the domains are equal, then equivalence
of the functions amounts to functionality checking, which is decidable
in \ptime.

Checking equivalence of the domains requires
some care, as $\Scal$ has partial updates. 
We have to prove that for every word $w\in \lanNFA{\Acal}$, the unique
run of $\Scal$ on $w$ is well-defined, meaning that the register
used in the final output function is defined in the associated valuation.
We claim that the domain of $\fctTrn{\Scal}$ is strictly included in 
$\lanNFA{\Acal}$ if, and only if, there exists a pair 
$(q,X)$ composed of a state $q$ and a register $X$ such that:
\begin{itemize}
\item $X$ is live in $q$, \ie there exists an input word $u$ from $q$
to some state $s$ such that $X$ flows along this execution to some register $Y$
and the output in $s$ is based on $Y$
\item there exists an input word $v$ such that the unique run from the initial
state on $v$ reaches state $q$ and the resulting configuration is such that
$X$ is undefined in the resulting valuation
\end{itemize}
It remains to explain how to decide these two properties.
The first point is just a problem of reachability, hence decidable in \ptime.
The second point is itself equivalent to the existence of a path as follows:
$$
q_0 \xrightarrow{v_1}q_1 \xrightarrow{\delta_1} q_2 \xrightarrow{v_2}q
$$
where $\delta_1$ is a transition with partial update, for which the update
of some register $X_1$ is not defined. Hence, the valuation is not defined
for $X_1$ after executing this transition. Then, $X_1$ flows to $X$
along the execution $q_2 \xrightarrow{v_2}q$.
Indeed, the only way for a variable to be undefined is from
a partial update.
Now, we can consider the flow graph, whose vertices are pairs 
composed of a state and a register. Using some classical trimming techniques,
we can consider all the transitions with partial updates that are reachable, and detect this way vertices of the flow graph which can be undefined this way. Then, using some saturation technique, which is in \ptime, we can deduce the set of 
vertices of the flow graph which can be undefined, solving the second point.

Hence, equivalence of the domains can also be checked in polynomial time.

\paragraph{\np-hardness}
We will show that we can adapt a reduction proposed in~\cite{Pfl1973}
from the graph coloring problem.
Let $G=(V,E)$ be a non-oriented graph, and some integer $k$.
We define an \SST $\Scal$ as follows. We define equivalently an asynchronous bimachine. Its left automaton has a single state, and its domain is universal.
The right automaton mimicks the construction presented in~\cite{Pfl1973}.
This construction defines a deterministic automaton, hence we mirror it
to build a codeterministic automaton.

Formally, it is defined as follows:
\begin{itemize}
\item Its set of states is $V\cup \{S_0,S_N,S_F\}$.
\item The alphabet is the set of vertices $V$.
\item $S_0$ is the unique final state, and all states are initial.
\item The transitions are defined as follows:
\begin{itemize}
\item for each vertex $v$, there is a transition $v\xrightarrow{v}S_0$, and a transition
$S_F\xrightarrow{v}v$,
\item for each pair of vertices $v,v'$, if $(v,v')\in E$, then we add a transition
$S_N \xrightarrow{v}v'$
\end{itemize}
\end{itemize}

The outputs are defined so as to distinguish between states $S_N$
and $S_F$.

Following the lines of the proof of~\cite{Pfl1973}, we can prove that the graph is $k$ 
colorable iff 
$\Scal$ can be realized with $k+3$ registers and the same underlying 
automaton. Indeed, the left automaton being trivial, using the correspondence
with asynchronous bimachines, the existence of an \SST $\Scal'$
 with the same underlying automaton is equivalent to that of a right automaton
 with $k+3$ states, which corresponds to being able to merge
 states in $V$ into $k$ classes.
\end{proof}





\subsection{Proof of Proposition~\ref{prop:pbRefEqPbExt}}
\begin{proof}
  For the first reduction, let $\rCRA \in \Nbb$ and let $\Tcal$ be a sequential
  letter-to-letter transducer (\resp a sequential transducer with $\dom(\fctTrn{\Tcal})$
  prefix-closed) with an underlying automaton $\Acal = (Q, q_0, \delta, F)$.

  Define by induction the \emph{production function}
  $\outSeqTr[\Tcal] \fct{Q \times \wds}{\wds[\alpp]}$ of $\Tcal$, for all $q \in Q$,
  by $q \outSeqTr[\Tcal] \e = \e$ and
  $q \outSeqTr[\Tcal] w \ltr = (q \outSeqTr[\Tcal] w)
    \oFct(q \cdot_\Acal w,\, \ltr,\, q \cdot_\Acal w\ltr)$
  for all $w \in \wds$ and $\ltr \in \alp$.
  Given a state $q$ and a word $w$, $q \outSeqTr[\Tcal] w$ is the output produced by $\Tcal$
  during its unique run on $w$ starting in $q$.

  Consider the instance of the minimal refinement problem where the precongruence
  on $\Acal$ given as an input is defined by
  \begin{equation}\label{eq:relCompat}
    p \approx q \text{ if and only if } p \outSeqTr[\Tcal] w = q \outSeqTr[\Tcal] w \text{ for all }
    w \in \wds
  \end{equation}

  If there exists a transducer $\Tcal'$ solution to the extension problem,
  then its underlying automaton $\Acal' = (Q', q_0', \delta', F')$ is a
  solution to the minimal refinement problem.

  Indeed, $\Acal'$ has at most $\rCRA$ states.
  Moreover, let $u,v \in \wds$ such that  $u \congDFA{\Acal'} v$
  (\ie $q_0' \cdot_{\Acal'} u = q_0' \cdot_{\Acal'} v$).
  For all $w \in \wds$, $\fctTrn{\Tcal}(uw) = \fctTrn{\Tcal'}(uw)$ and, since
  $\Tcal$ and $\Tcal'$ are letter-to-letter (\resp since $\fctTrn{\Tcal}$ is
  closed by prefix $u \in \dom(\fctTrn{\Tcal}) \subseteq \dom(\fctTrn{\Tcal'})$
  and thus $\fctTrn{\Tcal}(u) = \fctTrn{\Tcal'}(u)$), we have the equality
  \[
    (q_0 \cdot_\Acal u) \outSeqTr[\Tcal] w = (q_0 \cdot_{\Acal'} u) \outSeqTr[\Tcal'] w
  \]

  Of course, the same is also true for $v$ and thus
  \[
    (q_0 \cdot_\Acal u) \outSeqTr[\Tcal] w = (q_0 \cdot_{\Acal'} u) \outSeqTr[\Tcal'] w
    = (q_0 \cdot_{\Acal'} v) \outSeqTr[\Tcal'] w = (q_0 \cdot_\Acal v) \outSeqTr[\Tcal] w
  \]
  \ie $\congDFA{\Acal'}$ is finer than $\approx$.

  Conversely, if $\Acal' = (Q', q_0', \delta', F')$ is a solution
  to this instance of the minimal refinement problem then there exists a transducer
  $\Tcal' = (Q', \iniOFct', \oFct', \finOFct')$, with $\Acal'$ as its underlying
  automaton, solution to the extension problem.
  It is defined by
  \begin{itemize}
    \item $\iniOFct'(q_0') = \iniOFct(q_0)$.
    \item $\oFct'$ is defined, for all $u \in \wds$ and $\ltr \in \alp$, by
          $(q_0' \cdot_{\Acal'} u) \outSeqTr[\Tcal'] \ltr = (q_0 \cdot_\Acal u) \outSeqTr[\Tcal] \ltr$
          if $(q_0 \cdot_\Acal u) \outSeqTr[\Tcal] \ltr$ is defined and
          $(q_0' \cdot_{\Acal'} u) \outSeqTr[\Tcal'] \ltr = \ltr$ otherwise.
          It is well-defined since if $q_0' \cdot_{\Acal'} u = q_0' \cdot_{\Acal'} v$
          for some $ v\in \wds$ then $u \congDFA{\Acal'} v$ thus $u \approx v$ which
          means that $(q_0 \cdot_\Acal u) \outSeqTr[\Tcal] \ltr = (q_0 \cdot_\Acal v) \outSeqTr[\Tcal] \ltr$.
    \item $\finOFct'$ is defined, for all $u \in \wds$, by
          $\finOFct'(q_0' \cdot_{\Acal'} u) = \finOFct'(q_0 \cdot_\Acal u)$
          if $q_0 \cdot_\Acal u \in F$ and $\finOFct'(q_0' \cdot_{\Acal'} u) = \e$
          otherwise.
          For the same reasons, it is well-defined.
  \end{itemize}

  Clearly, for all $w \in \dom(\fctTrn{\Tcal})$, $\fctTrn{\Tcal'}(w) = \fctTrn{\Tcal}(w)$
  and $\Tcal'$ is indeed a sequential transducer with at most $\rCRA$ states
  realizing a total function.
  It is also letter-to-letter whenever $\Tcal$ is.

  For the converse reduction, let $\approx$ be a precongruence on the states
  of a \DFA $\Acal = (Q, q_0, \delta, F)$ and consider the instance of the
  extension problem where the transducer $\Tcal = (Q, \iniOFct, \oFct, \finOFct)$
  given as an input is defined as follows.

  For all $p,q \in Q$ such that $p \not\approx q$ define a letter
  $\ltr_{p,q} = \ltr_{q,p} \notin \alp$.
  Let $\alp' = \alp \cup \setWCnd{\ltr_{p,q}}{p,q \in Q, p \not\approx q}$ and
  let $\alpp = \alp \cup Q$.

  The transducer $\Tcal$ will realize a function $\fctTrn{\Tcal}
    \fct{\wds[\alp']}{\wds[\alpp]}$.
  Its underlying automaton will be $\Acal$ with additional self loops on the states
  $p,q \in Q$ such that $p \not\approx q$ on the letter $\ltr_{p,q}$.
  The output of the self loop on $p$ will be $p$ and the one on $q$ will be $q$,
  ensuring that no transducer realizing an extension of $\fctTrn{\Tcal}$ can
  merge the incompatible states $p$ and $q$ since they produce different outputs
  after reading the same letter $\ltr_{p,q}$.

  More formally, $\Tcal$ is defined by $\iniOFct(q_0) = \finOFct (q) = \e$
  for all $q \in F$ and $\oFct(p,\ltr,q) = \ltr$ for all $q = \delta (p, \ltr)$
  where $p \in Q$ and $\ltr \in \alp$ and $\oFct(p, \ltr_{p,q}, p) = p$ for all
  $p \in Q$ such that there exists $q \in Q$ such that $p \not\approx q$.

  Clearly, for all $p, q \in Q$, if $p \approx q$ then $p \outSeqTr[\Tcal] w = q \outSeqTr[\Tcal] w$
  for all $w \in \wds[\alp']$ and if $p \not\approx q$ then $ p \outSeqTr[\Tcal] \ltr_{p,q}
    \neq q \outSeqTr[\Tcal] \ltr_{p,q}$.
  Thus, the relation on the states of $\Tcal$ defined by~\eqref{eq:relCompat}
  is $\approx$ and the proof that the considered instance of the minimal refinement
  problem admits a solution if and only if this instance of the extension problem
  admits one is the same as in the previous reduction.
\end{proof}

\subsection{Proof of Theorem~\ref{thm:refinement}}
\begin{proof}
  Let $\Acal = (Q_\Acal, q_0, \delta_\Acal, F_\Acal)$, let $\approx$ and
  let $\Ecal = (Q_\Ecal, I_\Ecal, \Delta_\Ecal, F_\Ecal)$ be as the hypotheses
  of the theorem.

  Let $\Bcal = (Q_\Bcal, p_0, \delta_\Bcal, F_\Bcal)$
  be a \DFA finer than $\approx$ and define $\varphi \fct{Q_\Bcal}{Q_\Ecal}$,
  for all $p \in Q_\Bcal$, by $\varphi (p) = \setWCnd{q_0 \cdot_\Acal u}
    {u \in \wds \text{ such that } p_0\cdot_\Bcal u = p}$.
  It is well-defined since $\Bcal$ is finer than $\approx$.

  Define the \DFA $\Bcal' = (Q_{\Bcal'}, I_{\Bcal'}, \Delta_{\Bcal'}, F_{\Bcal'})$
  where
  \begin{itemize}
    \item $Q_{\Bcal'} = \varphi(Q_\Bcal)$
    \item $I_{\Bcal'} = \set{\varphi(p_0)}$
    \item for all $p \in Q_\Bcal$ and $\ltr \in \alp$,
          $\varphi(p) \cdot_{\Bcal'} \ltr = \varphi(p \cdot_\Bcal \ltr)$
    \item $F_{\Bcal'} = \setWCnd{\varphi(q)}{q \in F_\Bcal}$
  \end{itemize}

  Observe that $\congDFA{\Bcal}$ is finer than $\congDFA{\Bcal'}$ since
  $\varphi(p_0) \cdot_{\Bcal'} w = \varphi(p_0 \cdot_\Bcal w)$ for all $w \in \wds$.
  Note also that $\Bcal'$ is indeed a subautomaton of $\Ecal$ since $\varphi(p_0) \in I_\Ecal$
  because $q_0 \cdot_\Acal \e = q_0 \in \varphi(p_0)$ and, for all $p \in Q_\Bcal$,
  $\ltr \in \alp$ and $q \in \varphi(p)$, if $u \in \wds$ is such that
  $q_0 \cdot_\Acal u = q$ and $p_0 \cdot_\Bcal u = p$ then $p_0 \cdot_\Bcal u\ltr
    = p \cdot_\Bcal \ltr$ and $q_0 \cdot_\Acal u\ltr \in \varphi(p \cdot_\Bcal \ltr)$
  and thus $\big(\varphi(p), \ltr,  \varphi(p \cdot_\Bcal \ltr)\big) \in \Delta_\Ecal$.

  Moreover, $\card{Q_{\Bcal'}} \leq \card{Q_\Bcal}$.
  Consequently, in the case where $\Bcal$ is minimal, since $\Bcal'$ is also finer
  than $\approx$, $\card{Q_{\Bcal'}} = \card{Q_\Bcal}$ and $\varphi$ is a bijection.
\end{proof}

\begin{example}
  Let $\Acal$ be the \DFA depicted on Figure~\ref{fig:DFA_precong} and let
  $\approx$ be the precongruence on its states defined by $0 \approx 1$ and
  $0 \approx 2$ represented by dashed lines in the figure.

  \begin{figure}
    \savebox{\largest}{\tikzset{parallel/.style={auto=right,->,
      to path={ let \p1=(\tikztostart),\p2=(\tikztotarget),
          \n1={atan2(\y2-\y1,\x2-\x1)},\n2={\n1+180}
          in ($(\tikztostart.{\n1})!1mm!270:(\tikztotarget.{\n2})$) --
          ($(\tikztotarget.{\n2})!1mm!90:(\tikztostart.{\n1})$) \tikztonodes}}}

\scalebox{0.8}{
\begin{tikzpicture}[auto]
  \node[state, initial, accepting right] at (0,0) (0) {$\set{0}$};
  \node[state, accepting above] (1) at (3,2) {$\set{1}$};
  \node[state, accepting below] (2) at (3,-2) {$\set{2}$};
  \node[state, initial above, accepting right] (3) at (6.5,2) {$\set{0,2}$};
  \node[state, initial below, accepting right] (4) at (6.5,-2) {$\set{0,1}$};

  \draw[->]
  (0) edge[parallel] node {$a$} (1)
  (1) edge[parallel] node {$a$} (0)
  (0) edge[parallel] node {$b$} (2)
  (2) edge[parallel] node {$b$} (0)
  (1) edge[parallel] node {$b$} (2)
  (2) edge[parallel] node {$a$} (1)
  (1) edge[parallel] node {$a,b$} (3)
  (3) edge[parallel] node {$a$} (1)
  (2) edge[parallel] node {$a,b$} (4)
  (4) edge[parallel] node {$b$} (2)
  (3) edge[parallel] node {$a$} (4)
  (4) edge[parallel] node {$b$} (3)
  (1) edge node[near end, swap] {$a$} (4)
  (2) edge node[near end] {$b$} (3)
  (0) edge[bend left=60] node {$b$} (3)
  (0) edge[bend right=60] node [swap]{$a$} (4)
  (3) edge[loop right] node {$b$} (3)
  (4) edge[loop right] node {$a$} (4)
  ;

\end{tikzpicture}
}}
    \centering
    \mbox{}\hfill
    \subfloat[Automaton $\Acal$]{
      \raisebox{\dimexpr.5\ht\largest-.5\height}{
        \tikzset{parallel/.style={auto=right,->,
      to path={ let \p1=(\tikztostart),\p2=(\tikztotarget),
          \n1={atan2(\y2-\y1,\x2-\x1)},\n2={\n1+180}
          in ($(\tikztostart.{\n1})!1mm!270:(\tikztotarget.{\n2})$) --
          ($(\tikztotarget.{\n2})!1mm!90:(\tikztostart.{\n1})$) \tikztonodes}}}

\begin{tikzpicture}[auto, node distance=2cm]
  \node[state, initial, accepting right]  (0) {$0$};
  \node[state, accepting above] (1) [above right = of 0]{$1$};
  \node[state, accepting below] (2) [below right = of 0] {$2$};

  \draw[->]
  (0) edge[parallel] node {$a$} (1)
  (1) edge[parallel] node {$a$} (0)
  (0) edge[parallel] node {$b$} (2)
  (2) edge[parallel] node {$b$} (0)
  (1) edge[parallel] node {$b$} (2)
  (2) edge[parallel] node {$a$} (1)
  ;

  \draw[dashed]
  (0) edge[bend left = 50, dashed] (1)
  (0) edge[bend right = 50, dashed] (2)
  ;
\end{tikzpicture}}\label{fig:DFA_precong}
    }\hfill
    \subfloat[Automaton $\Ecal$]{
      \usebox{\largest}\label{fig:NFA_sub_ext}
    }\hfill\mbox{}\caption{}\label{fig:precong}
  \end{figure}

  The \NFA $\Ecal$ depicted on Figure~\ref{fig:NFA_sub_ext} is the automaton of the
  subset expansion of $\Acal$ relative to $\approx$.
  Its subautomaton corresponding to the states $\set{0,1}$ and $\set{0,2}$ is a
  complete \DFA respecting the compatibility relation imposed by $\approx$.
  It has fewer states than $\Acal$ and is, in fact, minimal among the \DFA finer
  than $\approx$.
  However, it is not unique as \eg the subautomaton of $\Ecal$ corresponding to
  the states $\set{2}$ and $\set{0,1}$ verifies the same property and has the
  same number of states.
\end{example}

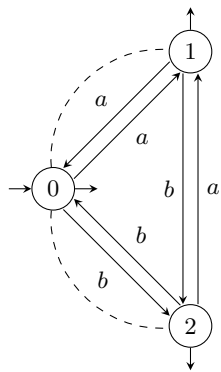
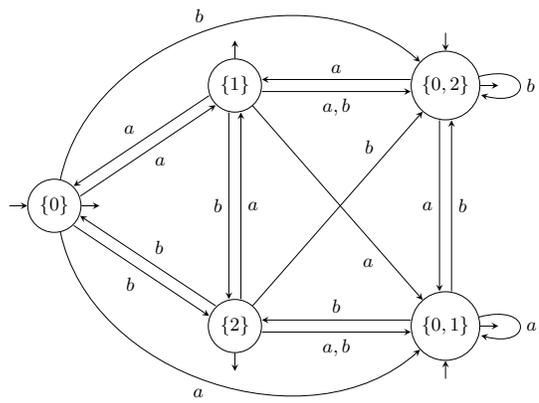
\end{document}